\documentclass[sigconf]{acmart}

\usepackage{xspace}
\usepackage[noend]{algpseudocode}

\setcopyright{none}
\acmConference[SPAA '19]{ACM Symposium on Parallelism in Algorithms and Architectures}{June 2019}{Phoenix, AZ, USA}
\acmYear{2019}
\copyrightyear{2019}

\setcitestyle{numbers,sort&compress}

\newcommand{\MultiCastCore}{\textsc{MultiCastCore}\xspace}
\newcommand{\MultiCast}{\textsc{MultiCast}\xspace}
\newcommand{\MultiCastC}{\textsc{MultiCast}$(C)$\xspace}
\newcommand{\MultiCastAdv}{\textsc{MultiCastAdv}\xspace}
\newcommand{\MultiCastAdvC}{\textsc{MultiCastAdv}$(C)$\xspace}

\begin{document}

\newtheorem{claim}{Claim}[theorem]


\title[{Fast and Resource Competitive Broadcast in Multi-channel Radio Networks}]{Fast and Resource Competitive Broadcast\\ in Multi-channel Radio Networks}

\author{Haimin Chen}
\orcid{}
\affiliation{Nanjing University}
\email{haimin.chen@smail.nju.edu.cn}

\author{Chaodong Zheng}
\orcid{}
\affiliation{Nanjing University}
\email{chaodong@nju.edu.cn}

\begin{abstract}
Consider a single-hop, multi-channel, synchronous radio network in which a source node needs to disseminate a message to all other $n-1$ nodes. An adversary called Eve, which captures environmental noise and potentially malicious interference, aims to disrupt this process via jamming. Assume sending, listening, or jamming on one channel for one time slot costs unit energy. The question is, if Eve spends $T$ units of energy on jamming, can we devise broadcast algorithms in which each node's cost is $o(T)$? Previous results show such \emph{resource competitive} algorithms do exist in the single-channel setting: each node can receive the message within $\tilde{O}(T+n)$ time slots while spending only $\tilde{O}(\sqrt{T/n}+1)$ energy.

In this paper, we show that when Eve is oblivious, the existence of multiple channels allows even \emph{faster} message dissemination, while preserving resource competitiveness. Specifically, we have identified an efficient ``epidemic broadcast'' scheme in the multi-channel setting that is robust again jamming. Extending this scheme leads to a randomized algorithm called \MultiCast which uses $n/2$ channels, and accomplishes broadcast in $\tilde{O}(T/n+1)$ time slots while costing each node only $\tilde{O}(\sqrt{T/n}+1)$ energy. When the value of $n$ is unknown, we further propose \MultiCastAdv, in which each node's running time is $\tilde{O}(T/(n^{1-2\alpha})+n^{2\alpha})$, and each node's cost is $\tilde{O}(\sqrt{T/(n^{1-2\alpha})}+n^{2\alpha})$. Here, $0<\alpha<1/4$ is a tunable parameter affecting the constant hiding behind the big-$O$ notation.

To handle the issue of limited channel availability, we have also devised variants for both \MultiCast and \MultiCastAdv that can work in networks in which only $C$ channels are available, for any $C\geq 1$. These variants remain to be resource competitive, and have (near) optimal time complexity in many cases.
\end{abstract}

\maketitle


\section{Introduction}\label{sec-intro}

Wireless networks are becoming increasingly popular over the last two decades. However, the open and shared nature of wireless medium makes these networks particularly vulnerable to various unintentional or even malicious interference. For example, microwave ovens could affect nearby devices which operate at the 2.4GHz band, and previous work has shown dedicated jammers are able to ``shutdown'' 802.11 networks easily~\cite{gummadi07}.

\bigskip\noindent\textbf{Resource competitive algorithms.} Over the years, researchers have proposed various solutions to mitigate these environmental and/or malicious interference (see related work section for more details), among which a recent and particularly interesting one is \emph{resource competitive} algorithms~\cite{king11,gilbert12,gilbert14,bender15,king18}.

Wireless devices are often battery powered, and sending/listening usually dominate the energy usage of these devices (e.g., sensors). On the other hand, there is a cost associated with causing interference as well. These observations motivate the notion of resource competitive algorithms: assume sending, listening, or jamming on one channel for one time slot costs one unit of energy,\footnote{In reality, the energy expenditure for sending, listening, and jamming might differ, but they are often in the same order. Our assumption of unit cost is standard and consistent with existing work (e.g., \cite{king11,gilbert12,gilbert14}). Moreover, allowing costs of different actions to be different constants will not affect the correctness of our results.} and assume there is an adversary Eve that is willing to spend $T$ units of energy to disrupt communication via jamming, can we devise algorithms that achieve certain goal, while maintaining each node's energy cost to be $o(T)$? Such algorithms would allow honest nodes to quickly ``bankrupt'' jammers, implying communication cannot be efficiently blocked.

Perhaps surprisingly, such algorithms do exist, even when Eve is adaptive. Specifically, King et al.~\cite{king18} show that in the 1-to-1 communication scenario, honest nodes only need to spend $O(\sqrt{T}+1)$ energy and can terminate in $O(T^2)$ time, in expectation. As for 1-to-$n$ broadcast, Gilbert et al.~\cite{gilbert14} show that all honest nodes can receive the message in $\tilde{O}(T+n)$ time, while spending only $\tilde{O}(\sqrt{T/n}+1)$ energy, with high probability.\footnote{Throughout the paper, we use $\tilde{O}$ to hide poly-log factors in $n$ and/or $T$; and we say an event happens ``\emph{with high probability (w.h.p.)}'' if it happens with probability at least $1-1/n^{\gamma}$, for some tunable constant $\gamma\geq 1$.} We note that all these results are obtained in the single-channel setting.

Modern radio devices are often able to access multiple channels (e.g., Bluetooth defines 79 channels, and 802.11ac defines over 100 channels); and government authorities are keeping releasing frequency bands for public use (e.g., TV white space spectrum~\cite{reardon10}). These trends raise an interesting question: will multiple channels allow \emph{faster} robust communication? Particularly, will resource competitive broadcast become faster in multi-channel radio networks?

In this paper, we answer the above question affirmatively, when Eve is an \emph{oblivious} adversary. That is, Eve knows the algorithm to be executed, and can pursue \emph{arbitrary} strategy, but she cannot observe algorithm execution and adjust her strategy accordingly. Admittedly, removing adaptivity simplifies analysis; but the core challenge of designing resource competitive algorithms---handling arbitrary adversary strategy---remains. We now detail our findings.

\bigskip\noindent\textbf{Epidemic broadcast and the \MultiCast algorithm.} To reduce the time complexity of broadcast, the most natural approach is to disseminate the message \emph{in parallel}. In distributed computing, this scheme is usually called ``epidemic broadcast'', as it allows the number of informed nodes (i.e., nodes which know the message) to grow exponentially, much like how a biological virus spreads. In a multi-channel radio network, message dissemination can happen on different channels in parallel: in each time slot, let each node independently choose a random channel, then let informed nodes broadcast and uninformed nodes listen. For each channel, so long as there is a single broadcaster and some listener, a message transmission succeeds. Note that this ``multi-channel epidemic broadcast'' scheme is also resource competitive against jamming: Eve cannot stop the number of informed nodes from increasing exponentially, unless she jams more than some constant fraction of all channels. For example, if $\Theta(n)$ channels are used, then in each time slot, cost of each node is only $O(1)$, but Eve has to spend $\Omega(n)$ energy to effectively disrupt broadcast.

\MultiCastCore is a simple and direct application of the above scheme. It is a randomized broadcast algorithm which uses $n/2$ channels, and ensures the runtime and energy consumption of each node is $\tilde{O}(T/n+1)$, with high probability. Unfortunately, \MultiCastCore needs $T$ as an input parameter, which is undesirable. Roughly speaking, this is because \MultiCastCore contains multiple identical iterations. To enforce correctness, the error probability of each iteration needs to be $O(1/T^{\Theta(1)})$. As a result, $T$ must be known in advance so as to set the iteration length properly.

To resolve this issue, we increase iteration length gradually as execution proceeds. With this modification, the error probability of an iteration naturally decreases as execution continues. We then apply another important adjustment: the broadcasting/listening probabilities of nodes also decrease as execution proceeds. Such ``sparse'' epidemic broadcast further improves competitiveness.

These changes lead to \MultiCast, an algorithm which also uses $n/2$ channels, and ensures the following properties with high probability: (a) all nodes receive the message and terminate within $\tilde{O}(T/n+1)$ time slots; and (b) the cost of each node is $\tilde{O}(\sqrt{T/n}+1)$. Notice, the energy expenditure matches the currently best known algorithm~\cite{gilbert14}, while the time consumption is significantly shorter. Thus, having multiple channels indeed allows faster message dissemination, without sacrificing resource competitiveness.

\bigskip\noindent\textbf{The \MultiCastAdv algorithm.} For ad hoc wireless networks, even knowledge of $n$ might be absent. Our third algorithm, called \MultiCastAdv, deals with such scenario. It contains multiple epochs, each of which contains multiple phases. In each phase, it makes a guess regarding the value of $n$, sets the number of channels to be used accordingly, and then executes an epidemic broadcast.

Several new challenges arise when designing \MultiCastAdv. First, we need to be more careful with setting broadcasting/listening probabilities. In particular, epidemic broadcast only works well in certain ``good'' phases. If honest nodes' energy expenditure in ``bad'' phases dominates the overall cost, resource competitiveness cannot be guaranteed, as Eve only needs to jam ``good'' phases.

The second challenge concerns with \emph{termination detection}, a particularly challenging and tricky issue in designing resource competitive algorithms. In epidemic broadcast, nodes need to stay around to help others even if they are informed. In \MultiCastCore and \MultiCast, a node halts when it hears few noisy slots over a time period, as this indicates low level of jamming, which in turn suggests epidemic broadcast must have succeeded. In \MultiCastAdv, by contrast, the number of used channels changes during execution, and low level of noise could also occur during ``bad'' phases in which epidemic broadcast fails. Thus, a more reliable criteria is hearing the message sufficiently often. However, this condition alone could result in other critical errors. Inspired by \cite{gilbert14}, we eventually adopt a two-stage termination mechanism so as to ensure both correctness and resource competitiveness.

Last but not least, we need to estimate $n$ to correctly identify ``good'' phases, otherwise resource competitiveness could again be broken. In our case, this goal cannot be achieved easily by observing simple metrics such as ``fraction of silent/message/noisy slots''. Instead, we craft a novel criterion via combining several metrics.

In the end, \MultiCastAdv guarantees the following properties with high probability: (a) all nodes receive the message and terminate within $\tilde{O}(T/(n^{1-2\alpha})+n^{2\alpha})$ slots; and (b) the cost of each node is $\tilde{O}(\sqrt{T/(n^{1-2\alpha})}+n^{2\alpha})$. Here, $0<\alpha<1/4$ is a tunable parameter. Notice, ideally $\alpha$ should be as small as possible, but the constant hiding behind the big-$O$ notation increases as $\alpha$ approaches zero.

\bigskip\noindent\textbf{Handling limited channel availability.} \MultiCast uses $n/2$ channels, and the number of channels required by \MultiCastAdv increases as the protocol proceeds. In real world, however, wireless spectrum is a scarce resource. To address this problem, we first describe a simple mechanism that can simulate \MultiCast when only $C\leq n/2$ channels are available. The resulting \MultiCastC algorithm accomplishes broadcast within $\tilde{O}(T/C+n/C)$ time slots, and each node's cost remain to be $\tilde{O}(\sqrt{T/n}+1)$. We then present \MultiCastAdvC, a variant of \MultiCastAdv that handles limited channel availability by a simple ``cut-off'' mechanism. In \MultiCastAdvC, the runtime of each node is dominated by the term $\tilde{O}\left(T/((\min\{C,n\})^{1-2\alpha})\right)$, and the cost of each node is dominated by the term $\tilde{O}\left(\sqrt{T/((\min\{C,n\})^{1-2\alpha})}\right)$, for any value of $C$.\footnote{Recall \MultiCast uses $n/2$ channels and needs $n$ as an input parameter, thus \MultiCastC only needs to handle the scenario in which $C\leq n/2$. In the case of \MultiCastAdvC, however, it needs to work for any value of $C$.} Notice, when $C=O(n)$, the runtime of \MultiCastC and \MultiCastAdvC are near optimal, as Eve can jam all $C$ channels for $T/C$ slots, blocking any communication.

\section{Related Work}\label{sec-related-work}

Broadcasting in radio networks is a well-studied problem. For example, early work from Bar-Yehuda et al.~\cite{bar-yehuda92} propose the widely used \textsc{Decay} procedure, and Alon et al.~\cite{alon91} establish the well-known $\Omega(\lg^2{n})$ time complexity lower bound. More recent results provide faster algorithms, some notable ones include \cite{kowalski05,ghaffari15,haeupler16,czumaj17}. Energy efficient broadcast has also been studied previously (see, e.g., \cite{berenbrink07,gasieniec07,chang18}), but often without considering malicious jamming.

When external interference is present, communication in wireless networks becomes more challenging. System researchers have proposed many physical layer and/or MAC layer approaches to detect jamming (e.g., \cite{xu05}), evade jamming (e.g., \cite{ma05,xu06b,navda07}), or even compete with jammers (e.g., \cite{xu06a}). The theory community, on the other hand, tend to focus on potential limitations the adversary may face and then develop corresponding countermeasures. For example, in an interesting series of papers, Awerbuch et al.~\cite{awerbuch08} and Richa et al.~\cite{richa10,richa11} study how to thwart adaptive jammers and reactive jammers in single-channel wireless networks by limiting the jamming rate.
When multiple channels are available, the restriction on the adversary is usually put on the number of channels she can disrupt simultaneously. Under such framework, Meier et al.~\cite{meier09} study how to solve the neighbor discovery problem efficiently. In another series of papers~\cite{dolev08,gilbert09,dolev10}, Dolev et al.\ and Gilbert et al.\ try to address gossiping in jamming-prone multi-channel wireless networks. Specifically, they have devised both deterministic and randomized algorithms.

These results provide valuable insights and interesting solutions to many important problems in challenging attack models. However, many of them would require nodes to incur significant cost. Moreover, the strategy Eve may employ is still limited: either she cannot jam continuously, or she cannot jam all channels simultaneously. Having observed this, Bender et al.~\cite{bender15} formalize and propose the notion of resource competitive analysis. In this framework, Eve can pursue arbitrary strategy, and the only limitation is her energy budget. This model better captures reality in many cases, but also poses new challenges to algorithm designers.

To the best of our knowledge, King, Saia, and Young~\cite{king11} propose the first resource competitive algorithm, in the context of 1-to-1 communication. (I.e., Alice wants to send a message to Bob.) Specifically, the devised Las Vegas algorithm ensures message delivery so long as Bob is correct. Moreover, cost of Alice and Bob is only $O(T^{0.62}+1)$ in expectation, while the adversary's expenditure is $T$. (The recent journal version~\cite{king18} provides a revised presentation, and serves as an excellent mini survey on resource competitive algorithms.) In \cite{gilbert12}, Gilbert and Young study 1-to-$n$ broadcast in which some nodes are Byzantine and controlled by an adversary. They devise a Monte Carlo resource competitive algorithm which ensures the message is delivered to most (but not all) correct nodes, with high probability. The work that is most closely related to ours is from Gilbert et al.~\cite{gilbert14}. In that paper, the problem in concern is again 1-to-$n$ broadcast, but all nodes are correct, and a single jamming adversary possessing $T$ energy is present. The authors propose a Monte Carlo algorithm which ensures all nodes can get the message in $\tilde{O}(T+n)$ time, while spending only $\tilde{O}(\sqrt{T/n}+1)$ energy, with high probability. The authors have also proved lower bounds to demonstrate nodes' energy expenditure are near optimal.

These pioneering works on resource competitive algorithms all focus on single-channel radio networks, and often assume the adversary is adaptive. In this paper, we consider multi-channel radio networks, assuming a weaker oblivious adversary is present.
We see our work as a first step in understanding how multiple channels affect the performance of resource competitive algorithms.

Lastly, we note that resource competitive analysis might also be relevant in other settings. Interested readers can refer to related work sections in \cite{bender15,king18} for more details.

\section{Model and Problem}\label{sec-model}

We consider a synchronous, single-hop, multi-channel radio network. In the network, there are $n$ \emph{honest nodes} (or simply \emph{nodes}), and one \emph{adversary} called Eve. For the ease of presentation, we assume $n$ is some power of two. Also, we often assume the number of available channels is unlimited when describing and analyzing our algorithms. Towards the end of the paper, we will discuss how to implement our algorithms with limited channel availability.

Time is divided into discrete \emph{slots}, and we assume all nodes start execution at the beginning of the same slot. In each time slot, each node can access one channel, and then choose to broadcast or listen on that channel, or remain idle. A node \emph{cannot} broadcast and listen simultaneously. For each node, the energy cost for broadcast and listen is one unit per slot, while idling incurs no cost. We assume all nodes can independently generate random bits.

In each slot, Eve can jam as many channels as she like, but she cannot spoof messages from honest nodes. Jamming one channel for one slot costs one unit of energy, and the total energy budget of Eve is $T$. We assume Eve is \emph{oblivious}: she knows the algorithm to be executed, and can pursue arbitrary strategy, but she cannot observe algorithm execution and adjust her strategy accordingly. She also does not know honest nodes' random bits.

For each channel, in each time slot, if no node broadcasts on this channel and Eve does not jam this channel, then every node listening on this channel will detect silence; if exactly one node $u$ broadcasts on this channel and Eve does not jam this channel, then every node listening on this channel will receive the message sent by $u$; lastly, if at least two nodes broadcast on this channel or Eve jams this channel (or both), then every node listening on this channel will hear noise. We assume nodes cannot distinguish between collision and jamming. Moreover, broadcasting nodes get no feedback regarding channel status.

The problem we are interested in is \emph{broadcast}, in which a single \emph{source node} wants to disseminate a message $m$ to all other nodes. During algorithm execution, we usually call a node \emph{informed} if it already knows $m$, otherwise the node is \emph{uninformed}.

We are interested in devising \emph{resource competitive} algorithms for the broadcast problem. Such algorithms should enforce two properties: (a) whenever there is no jamming or the jamming is weak, broadcast will soon succeed; and (b) during the time period of strong jamming, the energy cost of every honest node is much less than that of the adversary's. More formally, we adopt the following definition introduced by Bender et al.~\cite{bender15}.

\begin{definition}\label{def-resource-competitive-alg}
Consider an execution $\pi$ of an algorithm $\mathcal{A}$. Let $\texttt{cost}(\pi,u)$ denote the energy cost of an honest node $u$, and $T(\pi)$ denote the adversary's cost. We say $\mathcal{A}$ is \emph{$(\rho,\tau)$-resource competitive} if $\max_{u}\{\texttt{cost}(\pi, u)\}\leq\rho(T(\pi))+\tau$ for any execution $\pi$.
\end{definition}

In the above definition, $\rho$ is a function of Eve's cost and potentially other parameters. It captures the energy expenditure of an honest node when jamming is present, thus we want $\rho(T)\in o(T)$. On the other hand, $\tau$ denotes the unavoidable cost for accomplishing broadcast (in $\mathcal{A}$), even when Eve is absent (i.e., $T=0$). $\tau$ can be a function of parameters such as $n$, but it is not a function of $T$. Resource competitive algorithms usually focus on optimizing $\rho$.

\section{The \MultiCastCore Algorithm}\label{sec-multicastcore}

In this section, we present a simple algorithm called \MultiCastCore. It demonstrates some central ideas through which we achieve fast and resource competitiveness broadcast. The drawback of \MultiCastCore, however, is that it requires knowledge of $n$ and $T$.

As mentioned previously, a key integrant of \MultiCastCore is the epidemic broadcast scheme, in the multi-channel radio network setting. Nonetheless, to apply this scheme, we still need to decide how many channels to use in each time slot. Too few channels hurts parallelism, but too many channels may result in nodes not being able to meet each other sufficiently often, again reducing efficiency. As it turns out, $n/2$ channels is a good choice. To see this, let $t$ be the number of informed nodes. When $t\leq n/2$, we expect to see at least one uninformed node, and at most one informed node (thus informed nodes will not collide with each other), on each channel. Therefore, even if Eve jams constant fraction of all channels, the number of informed nodes can still increase by some constant factor in each slot. Once $t$ reaches $n/2$, we expect to see at least one, and at most some small constant number of informed nodes on each channel (thus contention among informed nodes is limited). Again, even if Eve jams constant fraction of all channels, all remaining uninformed nodes can quickly get the message.

Another key integrant of \MultiCastCore is termination detection: nodes need to decide when to halt. Clearly, a node should not terminate too late, as this increases energy expenditure. On the other hand, in epidemic broadcast, a node also should not terminate too early: parallel message dissemination cannot be achieved if informed nodes stop too soon. In \MultiCastCore, we use the number of noisy slots as the criteria: a node halts iff it does not observe too many noisy slots during one run of epidemic broadcast. The intuition is simple: if there are few noisy slots, then jamming from Eve cannot be strong, thus broadcast must have succeeded. Notice, as we elaborate later, using number of noisy slots is also critical for ensuring resource competitiveness.

We now describe \MultiCastCore in detail. \MultiCastCore contains multiple \emph{iterations}, each of which contains one run of epidemic broadcast. More specifically, let $\hat{T}=\max\{T,n\}$, then each iteration contains $R=a\lg{\hat{T}}$ slots, where $a$ is some sufficiently large constant. For each slot within an iteration, for each node $u$, it will go to a channel chosen uniformly at random from the range $[1,n/2]$. If $u$ is uninformed, then it will listen with probability $1/64$, and remain idle otherwise.\footnote{Throughout the paper, we do not attempt to optimize the constants, instead they are chosen for the ease of analysis and/or presentation.} If $u$ is informed, then it will listen or broadcast the message each with probability $1/64$, and remain idle otherwise. Node $u$ will also record the number of noisy slots it has observed within the current iteration. By the end of an iteration, $u$ will terminate if this number is less than $R/128$. The complete pseudocode of \MultiCastCore is shown in Figure \ref{fig-alg-multicastcore}.

\begin{figure}[!t]
\hrule
\vspace{1ex}\textbf{Pseudocode of \MultiCastCore executed at node $u$:}\vspace{1ex}
\hrule
\begin{normalsize}
\begin{algorithmic}[1]
\State $status\gets un$.
\If {(node $u$ is the source node)} $status\gets in$. \EndIf
\For {(each iteration)}
	\State $N_n\gets 0$.
	\For {(each slot from $1$ to $R=a\lg{\hat{T}}$)}
		\State $ch\gets\texttt{rnd}(1,n/2)$, $coin\gets\texttt{rnd}(1,64)$.
		\If {($coin==1$)}
			\State $feedback\gets\texttt{listen}(ch)$.
			\If {($feedback$ is noise)}
				\State $N_n\gets N_n+1$.
			\ElsIf {($feedback$ contains the message $m$)}
				\State $status\gets in$.
			\EndIf
		\ElsIf {($coin==2$ \textbf{and} $status==in$)}
			\State $\texttt{broadcast}(ch,m)$.
		\EndIf
	\EndFor
	\If {($N_n<R/128$)} \textbf{halt}. \EndIf
\EndFor
\end{algorithmic}
\end{normalsize}
\hrule\vspace{-1ex}
\caption{Pseudocode of the \MultiCastCore algorithm. (In the pseudocode: (a) $\texttt{rnd}(x,y)$ returns a uniformly chosen random integer from $[x,y]$; (b) $\texttt{listen}(ch)$ instructs the node to listen on channel $ch$ and returns channel status; and (c) $\texttt{broadcast}(ch,m)$ instructs the node to broadcast $m$ on $ch$.)}\label{fig-alg-multicastcore}
\vspace{-3ex}
\end{figure}

We now proceed to analyze \MultiCastCore. To begin with, we formally prove the effectiveness of the epidemic broadcast scheme, even when considerably amount of jamming is present.

\begin{lemma}\label{lemma-multicastcore-epidemic-bcst}
If an iteration satisfies: (a) all nodes are active at the beginning of it; and (b) for at least ten percent of all slots, Eve jams at most ninety percent of all $n/2$ channels. Then, by the end of this iteration, all nodes will be informed, with probability $1-1/\hat{T}^{\Omega(1)}$.
\end{lemma}

\begin{proof}
If an iteration satisfies (a) and (b), then in at least $2b\lg{\hat{T}}$ slots, at least $n/20$ channels are not jammed, where $b$ is some sufficiently large constant. Let $\mathcal{R}_1$ denote the collection of the first half of these $2b\lg{\hat{T}}$ slots, and let $\mathcal{R}_2$ denote the second half.

We first focus on $\mathcal{R}_1$. Consider an arbitrary slot in $\mathcal{R}_1$, let $t$ denote the number of informed nodes at the beginning of this slot. Define a channel to be \emph{good} if it is not jammed by Eve, and there is exactly one informed node broadcasting on that channel. Via the method of bounded differences (see, e.g., \cite{dubhashi09}), we can prove:\footnote{Due to space constraint, omitted proofs are provided in Appendix \ref{sec-appendix-proof}.}

\begin{claim}\label{claim-multicastcore-epidemic-bcst-1}
For any fixed slot in $\mathcal{R}_1$, if $t\leq n/2$, then the number of good channels is at least $\Theta(t)$, with at least constant probability.
\end{claim}

The above claim suggests, for each slot in $\mathcal{R}_1$, if $t\leq n/2$ at the beginning of it, then $t$ will increase by at least some constant factor by the end of this slot, with at least some constant probability. Since whether such exponential increase will happen are independent among the slots in $\mathcal{R}_1$, and since $|\mathcal{R}_1|=b\lg{\hat{T}}$, by a standard Chernoff bound (see, e.g., \cite{dubhashi09}), we know by the end of $\mathcal{R}_1$, the number of informed nodes will reach $n/2$, with probability $1-1/\hat{T}^{\Omega(1)}$.

We now turn attention to the slots in $\mathcal{R}_2$. Assume indeed the number of informed nodes is at least $n/2$ for all slots in $\mathcal{R}_2$. By an analysis similar to the proof of Claim \ref{claim-multicastcore-epidemic-bcst-1}, we know: for any fixed slot in $\mathcal{R}_2$, if $t\geq n/2$, then the number of good channels is $\Theta(n)$, with at least constant probability. Now, consider a slot in $\mathcal{R}_2$, and fix a node $u$ that is still uninformed at the beginning of this slot. By the end of this slot, $u$ will be informed with at least constant probability. Since $|\mathcal{R}_2|=b\lg{\hat{T}}$, we know $u$ will be informed by the end of $\mathcal{R}_2$, with probability $1-1/\hat{T}^{\Omega(1)}$. Take a union bound over all the $O(n)$ uninformed nodes, we know all nodes will be informed by the end of $\mathcal{R}_2$, with probability $1-1/\hat{T}^{\Omega(1)}$.
\end{proof}

Our next lemma demonstrates the correctness of the termination mechanism: when a node decides to terminate, all nodes must have been informed. That is, no node will halt before a successful epidemic broadcast is executed.

\begin{lemma}\label{lemma-multicastcore-halt-imply-inform}
Fix an iteration and a node $u$, assume all nodes are active at the beginning of this iteration. With probability $1-1/\hat{T}^{\Omega(1)}$, the following two events cannot happen simultaneously: (a) $u$ terminates by the end of this iteration; and (b) some node is still uninformed by the end of this iteration.
\end{lemma}

\begin{proof}
Let $\mathcal{E}_1$ be event (a), and $\mathcal{E}_2$ be event (b). Let $\mathcal{E}$ be the event that for at least ten percent of all slots within current iteration, Eve jams at most ninety percent of all $n/2$ channels. We know $\Pr(\mathcal{E}_1\mathcal{E}_2)=\Pr(\mathcal{E}_1\mathcal{E}_2\mathcal{E})+\Pr(\mathcal{E}_1\mathcal{E}_2\overline{\mathcal{E}})$.

Lemma \ref{lemma-multicastcore-epidemic-bcst} implies $\Pr(\mathcal{E}_1\mathcal{E}_2\mathcal{E})\leq\Pr(\mathcal{E}_2\mathcal{E})\leq\Pr(\mathcal{E}_2|\mathcal{E})\leq 1/\hat{T}^{\Omega(1)}$.

On the other hand, we bound $\Pr(\mathcal{E}_1\mathcal{E}_2\overline{\mathcal{E}})$ via $\Pr(\mathcal{E}_1|\overline{\mathcal{E}})$. If $\overline{\mathcal{E}}$ occurs, then Eve jams at least ninety percent of all $n/2$ channels for at least ninety percent of all slots within current iteration. Let $\mathcal{R}_1$ denote the set of slots in which Eve jams at least ninety percent of all $n/2$ channels. To make the number of noisy slots observed by $u$ as small as possible, without loss of generality, assume Eve jams exactly ninety percent of all $n/2$ channels for every slot in $\mathcal{R}_1$. Let $X_i$ be an indicator random variable taking value one iff $u$ listens on a channel jammed by Eve in the $i$\textsuperscript{th} slot in $\mathcal{R}_1$, where $1\leq i\leq|\mathcal{R}_1|$. Let $X=\sum_{i=1}^{|\mathcal{R}_1|}X_i$. We know $\mathbb{E}[X_i]=\Pr(X_i=1)=(9/10)\cdot(1/64)$, thus $\mathbb{E}[X]=|\mathcal{R}_1|\cdot\mathbb{E}[X_i]\geq (9R/10)\cdot\mathbb{E}[X_i]\geq R/80$. Now, notice that $\{X_1,X_2,\cdots,X_{|\mathcal{R}_1|}\}$ is a set of independent random variables, as nodes' channel choices and listening probabilities are not affected by Eve within an iteration, and we have assumed Eve jams exactly ninety percent of all $n/2$ channels for every slot in $\mathcal{R}_1$. Thus, by a Chernoff bound, $\Pr(X< R/128)$ is exponentially small in $R=a\lg{\hat{T}}$, implying $\Pr(\mathcal{E}_1~|~\overline{\mathcal{E}})\leq 1/\hat{T}^{\Omega(1)}$.
\end{proof}

We are not done yet. In particular, it could be the case that by the end of an iteration, all nodes are informed, but only \emph{some} of them decide to halt. Will remaining nodes ever terminate, and will \MultiCastCore remain competitive during this process?

The answer to both questions are yes, and this highlights another advantage of using fraction of noisy slots as the criteria for termination: \emph{decrease in the number of active nodes does not affect the ability of halting}, as less active nodes means less collisions, thus less noisy slots. As a result, if Eve wants to stop remaining active nodes from halting, she is again forced to spend much energy on jamming. The following lemma captures this observation:

\begin{lemma}\label{lemma-multicastcore-competitive-halt}
Fix an iteration and a node $u$, assume $u$ is active at the beginning of this iteration. If for at least eighty percent of all slots within this iteration, Eve jams at most twenty percent of all $n/2$ channels, then $u$ hears less than $R/128$ noisy slots within this iteration, with probability at least $1-1/\hat{T}^{\Omega(1)}$.
\end{lemma}

\begin{proof}
Let $\mathcal{R}_1$ be the set of slots in which Eve jams more than twenty percent of all $n/2$ channels, and let $\mathcal{R}_2$ be the remaining slots. Let $R_1=|\mathcal{R}_1|$, $R_2=|\mathcal{R}_2|$, we know $R_1\leq 0.2R$, $R_2\geq 0.8R$. To make the number of noisy slots observed by $u$ as large as possible, without loss of generality, assume: (a) Eve jams all channels for all slots in $\mathcal{R}_1$ and Eve jams exactly twenty percent of all channels for all slots in $\mathcal{R}_2$; (b) all nodes are active and informed at the beginning of this iteration; and (c) $R_1=0.2R$ and $R_2=0.8R$.

Let $X_i$ be an indicator random variable taking value one iff $u$ decides to listen in the $i$\textsuperscript{th} slot in $\mathcal{R}_1$, where $1\leq i\leq R_1$. We know $\mathbb{E}[X_i]=\Pr(X_i=1)=1/64$.

Let $Y_j$ be an indicator random variable taking value one iff $u$ hears noise in the $j$\textsuperscript{th} slot in $\mathcal{R}_2$, where $1\leq j\leq R_2$. Further define $J_j$ be an indicator random variable taking value one iff the channel chosen by $u$ in the $j$\textsuperscript{th} slot in $\mathcal{R}_2$ is jammed by Eve, and define $I_j$ be a random variable denoting the number of broadcasting informed nodes on the channel chosen by $u$ in the $j$\textsuperscript{th} slot in $\mathcal{R}_2$. We know $\Pr(Y_j=1)=(1/64)\cdot\left(\Pr(J_j=1)+\Pr(I_j\geq 2)\cdot\Pr(J_j=0)\right)=(1/64)\cdot\left(0.2+0.8\cdot\Pr(I_j\geq 2)\right)$. To upper bound $\Pr(I_j\geq 2)=1-\Pr(I_j=0)-\Pr(I_j=1)$, we give lower bounds for $\Pr(I_j=0)$ and $\Pr(I_j=1)$. In particular, $\Pr(I_j=0)=\left(1-(1/64)\cdot(2/n)\right)^{n-1}\geq \exp\left(-2\cdot(n-1)\cdot(1/64)\cdot(2/n)\right)\geq e^{-1/16}$; and $\Pr(I_j=1)=(n-1)\cdot(1/64)\cdot(2/n)\cdot\left(1-(1/64)\cdot(2/n)\right)^{n-2}\geq (n/2)\cdot(1/64)\cdot(2/n)\cdot e^{-1/16}=e^{-1/16}/64$. Hence, $\Pr(I_j\geq 2)\leq 1-(1+1/64)\cdot e^{-1/16}<0.05$, implying $\Pr(Y_j=1)<(1/64)\cdot0.24$.

Let $Z=\sum_{i=1}^{R_1}X_i + \sum_{j=1}^{R_2}Y_j$, we know $\mathbb{E}[Z]< R_1/64+(R_2/64)\cdot 0.24<R/160$. Notice that $\{X_1,X_2,\cdots,X_{R_1},Y_1,Y_2,\cdots,Y_{R_2}\}$ is a set of independent random variables. By a Chernoff bound, $\Pr(Z\geq R/128)$ is exponentially small in $R=a\lg{\hat{T}}$. Since $Z$ is an upper bound on the number of noisy slots $u$ will observe within current iteration, the lemma follows.
\end{proof}

We are now ready to state and prove the guarantees enforced by \MultiCastCore.

\begin{theorem}\label{thm-multicastcore}
When $n/2$ channels are available, the \MultiCastCore algorithm guarantees the following properties with high probability: (a) all nodes receive the message; and (b) the cost and active period of each node is $O(T/n+\max\{\lg{T},\lg{n}\})$.
\end{theorem}

\begin{proof}
Assume the last active node terminates by the end of iteration $l$. Let $L=\lfloor T/(0.02nR)\rfloor+1$. Let $A_u$ be the event that node $u$ is still active by the end of iteration $L$. Due to union bound, $\Pr(l>L)\leq \sum_{u}\Pr(A_u)=n\cdot\Pr(A_u)$. Notice that the total energy budge of Eve is $T$, thus among the first $L$ iterations, there must exist an iteration in which Eve spends less than $0.02nR$ energy. Assume iteration $\hat{i}$ is the first such iteration, and let $B$ denote the event that Eve spends less than $0.02nR$ energy in iteration $\hat{i}$, we have $\Pr(A_u)=\Pr(A_u\wedge B)$. But according to Lemma \ref{lemma-multicastcore-competitive-halt}, $\Pr(A_u\wedge B)\leq\Pr(A_u~|~B)\leq 1/\hat{T}^{\Omega(1)}$. (Specifically, if Eve spends less than $0.02nR$ energy in an iteration, then it must be the case that Eve jams at most twenty percent of all channels for at least eighty percent of all slots in that iteration.) Therefore, $\Pr(l>L)\leq 1/\hat{T}^{\Omega(1)}$. Since each iteration is of length $R=\Theta(\lg{\hat{T}})$, and since with probability $1-1/\hat{T}^{\Omega(1)}$ the energy cost of any fixed node in any fixed iteration is $\Theta(R)=\Theta(\lg{\hat{T}})$, we know all nodes will terminate within $O(T/n+\lg{\hat{T}})$ slots, and the cost of each node is also $O(T/n+\lg{\hat{T}})$, w.h.p.

Lastly, we note that Lemma \ref{lemma-multicastcore-halt-imply-inform} implies every node is informed when it decides to halt, w.h.p.
\end{proof}

Before proceeding to the next section, we note that \MultiCastCore possesses a nice property: once Eve stops disrupting protocol execution, all remaining active nodes will learn message $m$ (if they are still uninformed) and then halt, within one iteration. That is, within $\Theta(\lg{\hat{T}})=\Theta(\max\{\lg{T},\lg{n}\})$ slots. Existing resource competitive algorithms, including the other ones presented in this paper, usually demand at least $\tilde{\Theta}(T)$ slots for such scenario.

\section{The \MultiCast Algorithm}\label{sec-multicast}

In this section, we build upon \MultiCastCore and present \MultiCast, an algorithm which provides better resource competitiveness, and does not need $T$ as an input parameter.

To better understand \MultiCast, we first briefly discuss why \MultiCastCore needs to know $T$. In \MultiCastCore, all iterations are \emph{identical}. Since our algorithm is randomized and has small chance to fail, if $T$ is not available, the error probabilities of all iterations can only be expressed as a function of $n$. Thus, if $T$ is sufficiently large and the algorithm is executed long enough, bad events will eventually occur, with non-trivial probability. One simple solution to this problem is to let the error probability of each iteration also depends on $T$, resulting $T$ must be known at prior.

Inspired by previous work~\cite{king11,gilbert12,gilbert14,king18}, \MultiCast takes a more clever approach to solve this problem: make iterations \emph{different}. More specifically, we let the length of each iteration grow as the iteration number increases. In this way, later iterations are less likely to fail. Eventually, \MultiCast ensures the chance of error throughout the entire execution is bounded by some function of $n$, regardless of the length of the execution.

To make the algorithm more competitive, we have also decreased nodes' broadcasting/listening probabilities. Via a ``birthday paradox'' style analysis, we show this ``sparse'' epidemic broadcast is also correct. On the other hand, however, Eve still needs to jam at least constant fraction of all channels for at least constant fraction of all slots to effectively disrupt message dissemination.

With the above discussions in mind, we now describe \MultiCast in detail. The algorithm contains multiple iterations of different lengths. Specifically, the length of iteration $i$ is $R_i=a\cdot i\cdot 4^i\cdot \lg^2{n}$, for some sufficiently large constant $a$. For the ease of analysis, initially we set $i=6$. In each slot in iteration $i$, for each node $u$, it will go to a channel chosen uniformly at random from $[1,n/2]$. If $u$ is uninformed, then it will listen with probability $p_i=1/2^i$, and remain idle otherwise. If $u$ is informed, then it will listen or broadcast the message each with probability $p_i=1/2^i$, and remain idle otherwise. Node $u$ will also record the number of noisy slots it has observed within the current iteration. By the end of iteration $i$, node $u$ will terminate if this number is less than $R_ip_i/2$. The complete pseudocode of \MultiCast is shown in Figure \ref{fig-alg-multicast}.

\begin{figure}[!t]
\hrule
\vspace{1ex}\textbf{Pseudocode of \MultiCast executed at node $u$:}\vspace{1ex}
\hrule
\begin{normalsize}
\begin{algorithmic}[1]
\State $status\gets un$.
\If {(node $u$ is the source node)} $status\gets in$. \EndIf
\For {(each iteration $i\geq 6$)}
	\State $N_n\gets 0$.
	\For {(each slot from $1$ to $R_i=ai\cdot 4^i\cdot \lg^2{n}$)}
		\State $ch\gets\texttt{rnd}(1,n/2)$, $coin\gets\texttt{rnd}(1,2^i)$.
		\If {($coin==1$)}
			\State $feedback\gets\texttt{listen}(ch)$.
			\If {($feedback$ is noise)}
				\State $N_n\gets N_n+1$.
			\ElsIf {($feedback$ contains the message $m$)}
				\State $status\gets in$.
			\EndIf
		\ElsIf {($coin==2$ \textbf{and} $status==in$)}
			\State $\texttt{broadcast}(ch,m)$.
		\EndIf
	\EndFor
	\If {($N_n<R_i/2^{i+1}$)} \textbf{halt}. \EndIf
\EndFor
\end{algorithmic}
\end{normalsize}
\hrule\vspace{-1ex}
\caption{Pseudocode of the \MultiCast algorithm.}\label{fig-alg-multicast}
\vspace{-3ex}
\end{figure}

To analyze \MultiCast, we follow the same path as in the analysis of \MultiCastCore.

Once again, we begin by proving the effectiveness of the epidemic broadcast scheme, when jamming from Eve is limited. Notice, in the \MultiCast setting, we can no longer expect the number of informed nodes to grow exponentially in each time slot, as the broadcasting and listening probabilities of honest nodes are reduced. Instead, for the purpose of analysis, we divide an iteration into multiple segments each of length $\Theta(i\cdot 4^i\cdot\lg{n})$, and show that: (a) when the number of informed nodes is less than $n/2$, after each segment, each informed node will independently inform at least one uninformed node, so that the number of informed nodes will at least double; and (b) once the number of informed nodes reach $n/2$, all remaining uninformed nodes will be informed within a few segments. More precisely, we claim:

\begin{lemma}\label{lemma-multicast-epidemic-bcst}
If iteration $i$ satisfies: (a) all nodes are active at the beginning of it; and (b) for at least ten percent of all slots, Eve jams at most ninety percent of all $n/2$ channels. Then, by the end of this iteration, all nodes will be informed with probability at least $1-1/n^{5i}$.
\end{lemma}

Now that the effectiveness of the ``sparse'' epidemic broadcast scheme is established, by an analysis similar to the proof of Lemma \ref{lemma-multicastcore-halt-imply-inform}, we show \MultiCast correctly ensures no node will terminate before all nodes are informed:

\begin{lemma}\label{lemma-multicast-halt-imply-inform}
Fix an iteration $i$ and a node $u$, assume all nodes are active at the beginning of this iteration. With probability at least $1-1/n^{5i}$, the following two events cannot happen simultaneously: (a) $u$ terminates by the end of this iteration; and (b) some node is still uninformed by the end of this iteration.
\end{lemma}

Next, by an analysis similar to the proof of Lemma \ref{lemma-multicastcore-competitive-halt}, we claim once all nodes are informed and start halting, Eve cannot disrupt this process unless she spends a lot of energy.

\begin{lemma}\label{lemma-multicast-competitive-halt}
Fix an iteration $i$ and a node $u$, assume $u$ is active at the beginning of iteration $i$. If for at least eighty percent of all slots within this iteration, Eve jams at most twenty percent of all $n/2$ channels, then $u$ will terminate by the end of iteration $i$, with probability at least $1-1/n^{\Theta(i\cdot 2^i\cdot\lg{n})}$.
\end{lemma}

Finally, we state the guarantees enforced by \MultiCast in the following theorem. Due to space constraint, we only sketch the proof here. (The complete proof is provided in Appendix \ref{sec-appendix-proof}.)

\begin{theorem}\label{thm-multicast}
When $n/2$ channels are available, the \MultiCast algorithm guarantees the following properties with high probability: (a) all nodes receive the message and terminate within $O(T/n+\lg^2{n})$ slots; and (b) the cost of each node is $O(\sqrt{T/n}\cdot\sqrt{\lg{T}}\cdot\lg{n}+\lg^2{n})$.
\end{theorem}

\begin{proof}[Proof sketch]
Let $l$ be the last iteration in which Eve jams more than twenty percent of all the $n/2$ channels for more than twenty percent of all slots, it is easy to verify the cost of Eve in iteration $l$ is at least $0.02anl\cdot 4^l\cdot \lg^2{n}$, implying $T\geq 0.02anl\cdot 4^l\cdot \lg^2{n}$.

We first analyze the energy consumption of honest nodes. Fix a node $u$. For the first $l$ iterations, we show the total energy consumption of $u$ is at most $6al\cdot 2^l\cdot \lg^2{n}$, w.h.p. As for iteration $l+1$, we show the cost of $u$ will be at most $7al\cdot 2^{l}\cdot\lg^2{n}$, w.h.p. Notice Lemma \ref{lemma-multicast-competitive-halt} implies node $u$ will halt by the end of iteration $l+1$, w.h.p. Thus, node $u$'s total cost during the entire execution is at most $13al\cdot 2^{l}\cdot\lg^2{n}$, w.h.p. Let $E$ denote the total cost of node $u$, we know $E^2\leq (169al\lg^2{n})\cdot al\cdot 4^l\cdot\lg^2{n}\leq (169al\lg^2{n})\cdot T/(0.02n)=O(l\cdot (T/n)\cdot\lg^2{n})$. Moreover, it is easy to see $l=O(\lg{T})$. Thus, $E=O(\sqrt{T/n}\cdot\sqrt{\lg{T}}\cdot\lg{n})$. Take a union bound over all $n$ nodes, we know w.h.p., the total cost of each node is $O(\sqrt{T/n}\cdot\sqrt{\lg{T}}\cdot\lg{n})$.

We then analyze how long each node will remain active. Similar to the above analysis, we first show from the start of execution to the end of iteration $l$, the total time consumption is at most $2a\lg^2{n}\cdot l\cdot 4^l=O(T/n)$. On the other hand, the length of iteration $l+1$ is also $O(T/n)$. Apply Lemma \ref{lemma-multicast-competitive-halt}, we know all nodes will terminate by the end of iteration $l+1$, w.h.p. Therefore, all nodes will terminate within $O(T/n)$ slots, w.h.p.

We still need to show each node is informed when it terminates, and this is guaranteed by Lemma \ref{lemma-multicast-halt-imply-inform}.

Finally, we note that when Eve is not present (i.e., $T=0$) or $T=o(n)$, w.h.p.\ all honest nodes will be able to terminate by the end of the first iteration. In such scenario, the time and energy cost of each node is $O(\lg^2{n})$.
\end{proof}

\section{The \MultiCastAdv Algorithm}\label{sec-multicastadv}

We now present \MultiCastAdv, a resource competitive broadcast algorithm that further eliminates the dependency on $n$. The design and analysis of \MultiCastAdv is more involved than our previous two algorithms. Therefore, in the reminder of this section, we will first discuss the design decisions we made and the mechanisms we employed when building \MultiCastAdv, then give complete description of the algorithm, and finally proceed to the analysis.

\subsection{Crafting the Algorithm}

If $n$ is unknown, how can we determine the number of channels to be used? A natural approach is to \emph{guess} $n$. Imagine a protocol containing multiple \emph{epochs}, each of which contains multiple \emph{phases}. In the $j$\textsuperscript{th} phase within the $i$\textsuperscript{th} epoch (where $0\leq j\leq i-1$), nodes execute some variant of epidemic broadcast using $2^j$ channels, assuming $n\approx 2^{j+1}$. In this way, starting from epoch $\lg{n}$, there exists at least one ``good'' phase in each epoch in which the guess is correct (specifically, $j=\lg{n}-1$), and Eve must disrupt that phase to block message propagation.

This ``epoch-phase'' structure provides a skeleton for \MultiCastAdv, but we still need to fill in the details. To begin with, we must be careful with honest nodes' energy consumption of each phase. In particular, for honest nodes, ideally the total energy consumption of an epoch should be dominated by the energy consumption of phase $\lg{n}-1$ of that epoch. Otherwise, Eve could only jam phase $\lg{n}-1$ of each epoch, resulting in poor resource competitiveness. Unfortunately, we cannot directly prioritize phase $\lg{n}-1$ since $n$ in unknown. In \MultiCastAdv, our solution is to let the first phase of each epoch dominate the total energy consumption of that epoch, and try to decrease the energy consumption of each phase gently as $j$ increases. In this way, the energy consumption of phase $\lg{n}-1$ will be at least some sufficiently large fraction of the total energy consumption of current epoch.

For honest nodes, we also need to make sure that for each fixed $j$, the energy spent in the most recent phase $j$ dominates the sum of all energy spent in phase $j$ since the start of execution. That is, for each fixed $j$, if we extract the corresponding phases from all epochs, they should look similar to an execution of \MultiCast (using $2^j$ channels). This helps handle the problem that $T$ is unknown.

The next issue is termination detection. Recall in \MultiCastCore and \MultiCast, a node halts when it hears few noisy slots within an iteration. In \MultiCastAdv, this no longer works, as Eve can jam all phases until current phase number $j\gg\lg{n}$. In that phase, the number of channels used is too large and epidemic broadcast will fail, yet nodes will hear very few noisy slots (as nodes will not choose same channel frequently). Instead, the number of messages heard is a better criteria. In particular, in each phase, all nodes follow the same protocol and act independently. Therefore, if a node $u$ hears a lot of messages in a phase, then every other node is likely to have heard the message at least once, implying broadcast has succeeded and $u$ can safely terminate.

However, this simple criteria can result in critical error. Imagine Eve adjusts her strategy so as to let some nodes terminate first. Now, remaining nodes might never be able to halt, as there are not enough active nodes to broadcast messages! To resolve this issue, we take an approach inspired by Gilbert et al.~\cite{gilbert14}. Specifically, if a node $u$ observes sufficiently many message slots and silence slots in some phase $\hat{j}$ of some epoch $\hat{i}$, it does not halt immediately. Instead, it becomes a \emph{helper} and continues execution normally. Starting from epoch $\hat{i}+\Theta(1)$, in phase $\hat{j}$ within that epoch, node $u$ will check whether the number of noisy slots is sufficiently low. If indeed noisy slots are not frequent, then $u$ will halt. This two-stage termination mechanism enforces two nice properties: (a) all nodes must be informed when some node becomes helper; and (b) all nodes must be helper when some node halts. Critically, (b) implies the termination of some nodes will not affect the ability for remaining nodes to terminate: fewer active nodes means less noise.

Unfortunately, this two-stage termination mechanism creates a new problem: nodes need to become helper when the estimate of $n$ is nearly perfect (i.e., $2^{j+1}\approx n$). Otherwise, if some node $u$ obtains helper status in a phase $\hat{j}$ that is far from $\lg{n}$, then in later epochs, Eve only needs to jam phase $\hat{j}$ to prevent $u$ from halting, resulting in poor resource competitiveness. We resolve this issue by introducing another metric that nodes need to check before becoming helpers. As we detail in the analysis, this guarantees nodes will only become helper when $j=\lg{n}-1$.

\subsection{Algorithm Description}

The \MultiCastAdv algorithm contains multiple \emph{epochs}, each of which contains multiple \emph{phases}. In particular, for every epoch $i\geq 1$, there are $i$ phases within that epoch. These phases are numbered from $0$ to $i-1$. We use the term \emph{$(i,j)$-phase} to denote phase $j$ of epoch $i$. An $(i,j)$-phase contains two \emph{steps}, each of which contains $R(i,j)=\Theta(2^{2\alpha(i-j)}\cdot i^3)$ slots. Here, $0<\alpha<1/4$ is a tunable constant. We often use $R$ as a shorthand for $R(i,j)$ when the $(i,j)$ pair is clear from context. For each phase, the first step is mainly used for message dissemination, while the second step allows nodes to adjust their status and decide whether to halt.

Prior to algorithm execution, the status of the source node is set to \texttt{informed}, and the status of all other nodes are set to \texttt{uninformed}.

We now describe each step of an $(i,j)$-phase in detail. In each slot in the first step, each node will choose a channel uniformly at random from the range $[1,2^j]$ and hop to that channel. Then, each \texttt{uninformed} node will listen with probability $p(i,j)=2^{-\alpha(i-j)}/2$. (Again, we often use $p$ as a shorthand for $p(i,j)$ when the $(i,j)$ pair is clear from context.) If an \texttt{uninformed} node hears the message $m$, it becomes \texttt{informed}. On the other hand, if a node already knows $m$, then in each slot, it will broadcast $m$ with probability $p(i,j)$.

The second step of an $(i,j)$-phase is more complicated. At the beginning of the second step, each node will set four counters to zero: $N_m$, $N'_m$, $N_n$, and $N_s$. In each slot within step two, each node will again choose a channel uniformly at random from the range $[1,2^j]$ and hop to that channel. Each node will then choose to broadcast or listen each with probability $p(i,j)$. If a node chooses to broadcast, then it will broadcast the message $m$ if it is not \texttt{uninformed}, otherwise it will broadcast a special beacon message $\pm$. On the other hand, if a node chooses to listen, then it will increase $N_n$ (respectively, $N_s$) by one if it has observed a noisy (respectively, silent) slot. The tricky part is how $N_m$ and $N'_m$ are maintained: a node will increase $N_m$ if it hears the message $m$; and a node will increase $N'_m$ if it hears the message $m$ \emph{or} the special beacon message $\pm$. Also, the status of a node will \emph{not} change in the middle of step two. Critically, even if an \texttt{uninformed} node hears the message $m$ in a slot in step two, it will not change its status to \texttt{informed}.

By the end of step two, nodes will perform three checks, so as to adjust their status and decide whether to terminate. The first check is, if a node $u$ is currently \texttt{uninformed} yet $N_m\geq 1$, then it will update its status to \texttt{informed}. That is, if $u$ is \texttt{uninformed} prior to step two, but has learned the message $m$ during step two, then it will become \texttt{informed}. The second check is, if a node $u$ is currently \texttt{informed}, $N_m\geq 1.5Rp^2$, $N_s\geq 0.9Rp$, and $N'_m\leq 2.2Rp^2$, then it will update its status to \texttt{helper}. That is, if $u$ is currently \texttt{informed} and has heard $m$ and silence sufficiently often during step two, and if $2^j$ equals $n/2$ (which, as later analysis will show, can be deduced via values of $N'_m$ and $N_s$), then $u$ will become a \texttt{helper}. Moreover, when $u$ becomes a \texttt{helper}, it will record current $(i,j)$ pair as $(\hat{i}_u,\hat{j}_u)$. The last check concerns with termination. Specifically, if a node $u$ becomes \texttt{helper} in phase $(\hat{i}_u,\hat{j}_u)$, it will consider termination in phase $\hat{j}_u$ in epochs $i\geq \hat{i}_u+2/\alpha$. For each such phase, if during step two the number of noisy slots observed by $u$ is at most $Rp/3000$ (i.e., $N_n\leq Rp/3000$), then $u$ will halt by the end of that epoch.

See Figure \ref{fig-alg-multicastadv} in Appendix \ref{sec-appendix-figure} for complete pseudocode.

\subsection{Analysis}

We divide the analysis into several parts for clarity.

\bigskip\noindent\textbf{Good $(i,j)$ phases.} In this first part, we argue that when all nodes are active, \texttt{informed} nodes can only become \texttt{helper} when $i>\lg{n}$ and $j=\lg{n}-1$. This reduces the space of $(i,j)$ pairs we need to consider, and helps ensure the competitiveness of \MultiCastAdv.

We begin by showing that nodes cannot become \texttt{helper} when $i\leq\lg{n}$, as in these epochs too few channels are used and broadcasting probabilities are too high, thus collisions among nodes stop enough messages from being heard in step two of any phase.

\begin{lemma}\label{lemma-multicastadv-ij-range-1}
With high probability, no node will become \texttt{helper} during the first $\lg{n}$ epochs.
\end{lemma}

We then prove, nodes cannot progress to \texttt{helper} status in $(i,j)$-phases in which $j\geq\lg{n}$. Intuitively, this is because in such phases too many channels are being used, thus nodes cannot meet each other sufficiently often, implying $N_m$ cannot be large enough.

\begin{lemma}\label{lemma-multicastadv-ij-range-2}
Fix an $(i,j)$-phase in which $i>\lg{n}$ and $j\geq\lg{n}$. Fix a node $u$. With probability at least $1-n^{-\Theta(i^2)}$, node $u$ will not become \texttt{helper} by the end of this phase.
\end{lemma}

The last claim in this part is, when $i>\lg{n}$ and all nodes are active, nodes cannot progress to \texttt{helper} status in phases in which $j<\lg{n}-1$, as in such scenario $N'_m\leq 2.2Rp^2$ and $N_s\geq 0.9Rp$ cannot be satisfied simultaneously.

\begin{lemma}\label{lemma-multicastadv-ij-range-3}
Fix an $(i,j)$-phase in which $i>\lg{n}$, $j< \lg n-1$, and all nodes are active. Fix a node $u$. With probability at least $1-n^{-\Theta(i^2)}$, node $u$ will not become \texttt{helper} by the end of this phase.
\end{lemma}

\smallskip\noindent\textbf{Correctness and competitiveness guarantees.} In this part, we show \MultiCastAdv enforces two nice properties: (a) when some node becomes \texttt{helper}, all nodes must have learned the message $m$; and (b) when some node decides to terminate, all other nodes must have at least progressed to \texttt{helper} status.

We begin with the first property. Consider a node $u$. If $u$ becomes \texttt{helper} in phase $\hat{j}_u$ of epoch $\hat{i}_u$, then it must have heard message $m$ sufficiently often during step two of that phase. As a result, the expectation of $N_m$ must be sufficiently large in phase $(\hat{i}_u,\hat{j}_u )$. On the other hand, due to Lemma \ref{lemma-multicastadv-ij-range-1}, we can restrict our attention to epochs where $i>\lg{n}$, this in turn allows us to use concentration inequalities to show each node's observed $N_m$ will be close to $\mathbb{E}[N_m]$, w.h.p. Therefore, we know all nodes must have learned $m$ by the end of phase $(\hat{i}_u,\hat{j}_u )$, w.h.p. More precisely:

\begin{lemma}\label{lemma-multicastadv-helper-imply-inform}
Fix an $(i,j)$-phase in which $i>\lg{n}$, by the end of step two, the following two events happen simultaneously with probability at most $n^{-\Theta(i^2)}$: (a) some node has $N_m\geq 1.5Rp^2$; and (b) some node still does not know the message $m$.
\end{lemma}

Our next key technical lemma states, when nodes begin to terminate, all nodes must have at least progressed to \texttt{helper} status. As discussed earlier, this property ensures the termination of some nodes will not affect the ability for remaining nodes to halt. To prove this lemma, we focus on the \emph{first} node that decides to terminate. Call this node $u$. Recall our algorithm specifies that if $u$ becomes \texttt{helper} in phase $(\hat{i}_u,\hat{j}_u)$, then node $u$ will only consider termination in $(i,j)$-phases in which $i\geq \hat{i}_u+\Theta(1)$ and $j=\hat{j}_u$. Now, assume $u$ decides to terminate in some phase $(\tilde{i},\hat{j}_u)$, then in step two of the earlier phase $(\hat{i}_u,\hat{j}_u)$ in which $u$ became \texttt{helper}, it must have observed message $m$ slots and silent slots sufficiently often. Notice that as the protocol proceeds, in step two of $(i,j)$-phases in which $j=\hat{j}_u$, if Eve does no jamming, fraction of message $m$ slots and silent slots both increase. Thus in step two of phase $(\tilde{i},\hat{j}_u)$, if jamming from Eve is not strong, all nodes should have heard many message $m$ slots and silent slots. But since $u$ decides to terminate in phase $(\tilde{i},\hat{j}_u)$, it must have heard few noisy slots, implying jamming from Eve is indeed weak. Lastly, notice that our previous analysis ensures $\hat{j}_u$ must equal to $\lg{n}-1$, implying $N'_m$ cannot be too large in phase $(\tilde{i},\hat{j}_u)$. Combine all these observations and we can conclude, all nodes will become \texttt{helper} in phase $(\tilde{i},\hat{j}_u)$.

\begin{lemma}\label{lemma-multicastadv-halt-imply-helper}
Fix an $(i,j)$-phase in which all nodes are active, fix a node $u$ that became \texttt{helper} in phase $(\hat{i},\hat{j})$. Assume $i-2/\alpha\geq \hat{i}>\lg{n}$. By the end of step two, the following two events happen simultaneously with probability at most $n^{-\Theta(\hat{i}^2)}$: (a) node $u$ decides to halt; and (b) some node has not progressed to \texttt{helper} status.
\end{lemma}

\smallskip\noindent\textbf{Fast termination.} As the final preparation before proving the main theorem, we show that once the Eve stops disrupting protocol execution (i.e., jamming from Eve is not strong), remaining active nodes will quickly learn message $m$ and then terminate.

We first introduce some notations. For any $(i,j)$-phase, we use $\mathcal{E}^{> x}_{Step1}(> y)$ (respectively, $\mathcal{E}^{> x}_{Step2}(> y)$) to denote the event that during step one (respectively, step two) of phase $(i,j)$, Eve jams more than $y$ fraction of all $2^j$ channels for more than $x$ fraction of all slots within that step. Essentially, the stronger Eve jams, the bigger $x$ and $y$ are. It is easy to see the negation of $\mathcal{E}^{> x}_{Step1}(> y)$ is $\mathcal{E}^{\geq 1-x}_{Step1}(\leq y)$, and the negation of $\mathcal{E}^{> x}_{Step2}(> y)$ is $\mathcal{E}^{\geq 1-x}_{Step2}(\leq y)$.

We now classify epochs into \emph{blocking} and \emph{non-blocking}:

\begin{definition}\label{def-epoch-blocking}
Epoch $i$ is \emph{blocking} if at least one of the following two conditions hold: (a) $\mathcal{E}^{> x_1}_{Step1}(> y_1)$ in phase $\lg{n}-1$; and (b) $\mathcal{E}^{> x_2}_{Step2}(> y_2)$ in phase $\lg{n}-1$. Here, $x_1=y_1=1/10$, $x_2=y_2=1/10^4$. On the other hand, epoch $i$ is \emph{non-blocking} if both of the following conditions hold: (a) $\mathcal{E}^{\geq 1-x_1}_{Step1}(\leq y_1)$ in phase $\lg{n}-1$; and (b) $\mathcal{E}^{\geq 1-x_2}_{Step2}(\leq y_2)$ in phase $\lg{n}-1$.
\end{definition}

We are ready to prove the fast termination properties enforced by \MultiCastAdv, and we begin by showing that if there are \texttt{uninformed} nodes at the beginning of a non-blocking epoch, then all of them will learn message $m$ during this epoch, as a successful epidemic broadcast will happen during phase $\lg{n}-1$.

\begin{lemma}\label{lemma-multicastadv-fast-term-inform}
Consider phase $\lg{n}-1$ of an epoch $i>\lg{n}$, assume all nodes are active and are either \texttt{uninformed} or \texttt{informed} at the beginning of this phase. If $\mathcal{E}^{\geq 1-x_1}_{Step1}(\leq y_1)$ happens, then by the end of this phase, all nodes are \texttt{informed} with probability at least $1-n^{-\Theta(i)}$.
\end{lemma}

Next, we show that if all nodes are \texttt{informed} or \texttt{helper}, then after a non-blocking epoch, all nodes must have progressed to \texttt{helper} status or have terminated.

\begin{lemma}\label{lemma-multicastadv-fast-term-helper}
Consider phase $\lg{n}-1$ of an epoch $i\geq\lg{n}+6/\alpha$, assume all nodes are active and are either \texttt{informed} or \texttt{helper} at the beginning of step two of this phase. If $\mathcal{E}^{\geq 1-x_2}_{Step2}(\leq y_2)$ happens, then by the end of this phase, all nodes must be in \texttt{helper} or \texttt{halt} status, with probability at least $1-n^{-\Theta(i^2)}$.
\end{lemma}

The last lemma demonstrates, once a \texttt{helper} starts considering termination, after $O(1)$ epochs, the node will successfully halt in a non-blocking epoch. More precisely, we claim:

\begin{lemma}\label{lemma-multicastadv-fast-term-halt}
Consider a node $u$ that obtained \texttt{helper} status in phase $(\hat{i},\hat{j})$ in which $\hat{i}>\lg{n}$ and all nodes are active. Consider phase $j=\hat{j}$ of an epoch $i\geq \hat{i}+11/\alpha$ in which $u$ is active, if $\mathcal{E}^{\geq 1-x_2}_{Step2}(\leq y_2)$ happens, then by the end of this phase, $u$ will be in \texttt{halt} status, with probability at least $1-n^{-\Theta(\hat{i}^2)}$.
\end{lemma}

\smallskip\noindent\textbf{Putting things together.} We are now ready to state the main theorem and sketch its proof. (Once again, the complete proof is provided in Appendix \ref{sec-appendix-proof}.)

\begin{theorem}\label{thm-multicastadv}
The \MultiCastAdv algorithm guarantees the following properties with high probability: (a) all nodes receive the message; (b) all nodes terminate within $O(T/(n^{1-2\alpha})\cdot\lg^3{T}+n^{2\alpha}\cdot\lg^3{n})=\tilde{O}(T/(n^{1-2\alpha})+n^{2\alpha})$ time slots; and (c) the cost of each node is $O(\sqrt{T/(n^{1-2\alpha})}\cdot\lg^3{T}+n^{2\alpha}\cdot\lg^3{n})=\tilde{O}(\sqrt{T/(n^{1-2\alpha})}+n^{2\alpha})$.
\end{theorem}

\begin{proof}[Proof sketch]
To begin with, we prove the following claims: (a) due to Lemma \ref{lemma-multicastadv-ij-range-1} and Lemma \ref{lemma-multicastadv-helper-imply-inform}, when the first \texttt{helper} appears, all nodes must have learned the message $m$, w.h.p.; (b) due to Lemma \ref{lemma-multicastadv-ij-range-1} and Lemma \ref{lemma-multicastadv-halt-imply-helper}, when the first \texttt{halt} node appears, all nodes must have at least progressed to \texttt{helper} status, w.h.p.; and (c) due to Lemma \ref{lemma-multicastadv-ij-range-1}, Lemma \ref{lemma-multicastadv-ij-range-2}, and Lemma \ref{lemma-multicastadv-ij-range-3}, for each node $u$, if it obtained \texttt{helper} status in phase $(\hat{i}_u,\hat{j}_u)$, then it must be the case that $\hat{i}_u>\lg{n}$ and $\hat{j}_u=\lg{n}-1$, w.h.p.

Let $l$ denote the last epoch that is blocking while some node is still active at the beginning of it. Due to the above analysis, by the end of epoch $l$, the status of all nodes must belong to exactly one of the following four cases: (1) all nodes are active and either \texttt{uninformed} or \texttt{informed} (and there exists at least one \texttt{uninformed} node); (2) all nodes are active and either \texttt{informed} or \texttt{helper} (and there exists at least one \texttt{informed} node); (3) every node is either \texttt{helper} or has terminated (and there exists at least one \texttt{helper} node); or (4) all nodes have terminated.

Following analysis consider two scenarios: $l\geq\lg{n}$ and $l<\lg{n}$.

\textit{The $l\geq\lg{n}$ scenario.} In this situation, due to to Lemma \ref{lemma-multicastadv-fast-term-inform}, Lemma \ref{lemma-multicastadv-fast-term-helper}, and Lemma \ref{lemma-multicastadv-fast-term-halt}, we show no matter which case (among the aforementioned four cases) the system is in by the end of epoch $l$, all nodes will halt by the end of epoch $l+18/\alpha$, w.h.p.

Thus, to bound the runtime and cost of honest nodes, we need to bound $l$. Since epoch $l$ is blocking, we can show Eve's cost during epoch $l$ is at least $\Theta(1)\cdot x_2y_2\cdot 2^{2\alpha l}\cdot l^3\cdot 2^{(1-2\alpha)(\lg{n}-1)}$, implying $T\geq (x_2y_2/2)\cdot 2^{2\alpha l}\cdot l^3\cdot n^{(1-2\alpha)}\geq (x_2y_2/2)\cdot 2^{2\alpha l}\cdot n^{(1-2\alpha)}$. As a result, $l\leq (\lg(T/((x_2y_2/2)\cdot n^{(1-2\alpha)})))/(2\alpha)$.

Therefore, each node will halt within $\sum_{i=1}^{l+18/\alpha}\sum_{j=0}^{i-1}\Theta(1)\cdot 2^{2\alpha i}\cdot 2^{-2\alpha j}\cdot i^3=O(T/(n^{(1-2\alpha)})\cdot\lg^3{T})$ slots. On the other hand, we show the total energy consumption of a node is at most the sum of $\sum_{i=1}^{\lg{n}}\sum_{j=0}^{i-1}\Theta(1)\cdot 2^{2\alpha i}\cdot 2^{-2\alpha j}\cdot i^3$ and $\sum_{i=\lg{n}+1}^{l+18/\alpha}\sum_{j=0}^{i-1}\Theta(1)\cdot 2^{\alpha i}\cdot 2^{-\alpha j}\cdot i^3$, implying the total cost of each node is $O(\sqrt{T/(n^{(1-2\alpha)})}\cdot\lg^3{T}+n^{2\alpha}\cdot\lg^3{n})$, with high probability.

\textit{The $l<\lg{n}$ scenario.} This situation is easier to analyze. By an analysis similar to the $l\geq\lg{n}$ scenario, we show all nodes will terminate by the end of epoch $\lg{n}+18/\alpha$, with high probability. Hence, the total runtime and cost for each node in this situation is at most $\sum_{i=1}^{\lg{n}+18/\alpha}\sum_{j=0}^{i-1}\Theta(1)\cdot 2^{2\alpha i}\cdot 2^{-2\alpha j}\cdot i^3=O(n^{2\alpha}\cdot\lg^3{n})$.

We conclude the proof by noting that the $l<\lg{n}$ scenario also includes the situation in which Eve is not present.
\end{proof}

\section{Handling Limited Channel Availability}

Previous algorithms might need $\Omega(n)$ channels, but wireless spectrum is a scarce resource. We now discuss how to implement our algorithms when limited number of channels are available.

Consider a multi-channel algorithm $\mathcal{A}$. For each slot $s$ during the execution of $\mathcal{A}$, for each node $u$, let $\mathcal{C}(s,u)$ denote the set of channels that $\mathcal{A}$ may instruct $u$ to use in slot $s$. We say $\mathcal{A}$ is \emph{channel-uniform} if for every slot $s$, $\mathcal{C}(s,u)$ are identical for all nodes that are still active, and we use $\mathcal{C}(s)$ to denote $\mathcal{C}(s,u)$. Clearly, the three algorithms introduced so far are all channel-uniform: for \MultiCastCore and \MultiCast, $\mathcal{C}(s)$ is always $[1,n/2]$; and for each slot in phase $(i,j)$ of \MultiCastAdv, $\mathcal{C}(s)$ is $[1,2^j]$.

There exists a simple mechanism to simulate any channel-uniform algorithm $\mathcal{A}$ in a radio network $\mathcal{N}_C$ in which only $C$ channels are available. Specifically, given $\mathcal{A}$, we build $\mathcal{A}_C$ that works in $\mathcal{N}_C$ in the following way. $\mathcal{A}_C$ contains multiple \emph{rounds}, each of which is used to simulate one slot of $\mathcal{A}$. In particular, for every slot $s$ in $\mathcal{A}$, the corresponding round in $\mathcal{A}_C$ contains $\lceil |\mathcal{C}(s)|/C\rceil$ slots. For every node that is active in slot $s$ in $\mathcal{A}$, if it decides to access the $k$\textsuperscript{th} channel in $\mathcal{C}(s)$ in $\mathcal{A}$, then it will use channel $((k-1)\mod C)+1$ in slot $\lfloor (k-1)/C\rfloor+1$ in the corresponding round in $\mathcal{A}_C$.

If there is no adversary or environmental interference, clearly $\mathcal{A}_C$ can perfectly simulate $\mathcal{A}$, with some overhead on running time. In fact, in the context of resource competitive broadcast, applying this simple simulation mechanism to \MultiCastCore and \MultiCast also preserves resource competitiveness.

Concretely, we now describe how to adjust \MultiCast to work in a radio network in which only $C\leq n/2$ channels are available. Call this variant \MultiCastC. For the ease of presentation, assume $n/2$ is some multiple of $C$. (Otherwise, round down $C$.) Similar to \MultiCast, \MultiCastC contains multiple iterations, where the $i$\textsuperscript{th} iteration contains $R_i=\Theta(i\cdot 4^i\cdot \lg^2{n})$ rounds. (Again, initially we set $i=6$.) Recall each round simulates one slot of \MultiCast, thus each round in \MultiCastC contains $n/(2C)$ slots, implying the actual length of the $i$\textsuperscript{th} iteration in \MultiCastC is $\Theta((n/C)\cdot i\cdot 4^i\cdot \lg^2{n})$ slots. The pseudocode of \MultiCastC is provided in Figure \ref{fig-alg-multicastc} in Appendix \ref{sec-appendix-figure}.

The analysis of \MultiCastC is essentially identical to that of the \MultiCast's. Specifically, by substituting ``slots'' with ``rounds'' in the statements of Lemma \ref{lemma-multicast-epidemic-bcst}, \ref{lemma-multicast-halt-imply-inform}, and \ref{lemma-multicast-competitive-halt}, the proofs still hold. Moreover, going through the proof of Theorem \ref{thm-multicast} with updated iteration length immediately leads to the following corollary:

\begin{corollary}\label{corollary-multicastc}
When $1\leq C\leq n/2$ channels are available, the \MultiCastC algorithm guarantees the following properties with high probability: (a) all nodes receive the message and terminate within $O(T/C+(n/C)\cdot\lg^2{n})$ slots; and (b) the cost of each node is $O(\sqrt{T/n}\cdot\sqrt{\lg{T}}\cdot\lg{n}+\lg^2{n})$.
\end{corollary}

Applying this simulation mechanism to \MultiCastAdv, however, may result in poor time consumption (though resource competitiveness is preserved). Roughly speaking, this is because simulating phase $j$ when $j\gg\lg{n}$ costs too much time, yet Eve only needs to focus on jamming phases in which $j=\lg{n}-1$.

Interestingly, \MultiCastAdv is robust enough so that with a simple ``cut-off'' mechanism, we can devise a variant called \MultiCastAdvC that preserves competitiveness and has desirable running time. Specifically, each epoch $i$ of \MultiCastAdvC is just like epoch $i$ of \MultiCastAdv, except that it does not contain phases in which $j>\lg{C}$, if such phases exist in epoch $i$. We also make a small adjustment to nodes' behavior in phases where $j=\lg{C}$: in each such phase, for each node, if it is currently \texttt{informed}, then it will become \texttt{helper} if $N_m\geq 1.5Rp^2$ and $N_s\geq 0.9Rp$ (i.e., we remove the condition $N'_m\leq 2.2Rp^2$). See Figure \ref{fig-alg-multicastadvc} in Appendix \ref{sec-appendix-figure} for the pseudocode of \MultiCastAdvC.

We now discuss the performance of \MultiCastAdvC, and we consider two complement cases: $C>n/2$ and $C\leq n/2$.

When $C>n/2$, our analysis of \MultiCastAdv is still correct in the \MultiCastAdvC setting. This is because, in \MultiCastAdv, the ``good'' phases we really care about are the ones in which $j=\lg{n}-1$, but when $C>n/2$, these phases also exist in \MultiCastAdvC. Therefore, if $C>n/2$, \MultiCastAdvC provides the same guarantee as \MultiCastAdv: Theorem \ref{thm-multicastadv} still holds.

When $C\leq n/2$, the situation becomes a bit more complicated. To begin with, again nodes can only obtain \texttt{helper} status in epochs $i>\lg{n}$. Moreover, \MultiCastAdvC retains the nice properties that: (a) when some \texttt{helper} appears, all nodes have already learned the message; and (b) when some node \texttt{halt}, all nodes have at least obtained \texttt{helper} status. However, now nodes will not become \texttt{helper} in phases in which $j=\lg{n}-1$. Instead, this status change will happen when $j=\lg{C}$. Lastly, the ``fast termination'' properties are also affected slightly: nodes need some more time to reach \texttt{halt} status once Eve ceases jamming, but the impact to performance is limited. Due to space constraint, see Appendix \ref{sec-appendix-multicastadvc} for a more detailed analysis of \MultiCastAdvC when $C\leq n/2$, here we only state the guarantees enforced by it.

\begin{theorem}\label{thm-multicastadvc}
When $1\leq C\leq n/2$ channels are available, the \MultiCastAdvC algorithm guarantees the following properties with high probability: (a) all nodes receive the message; (b) all nodes terminate within $\tilde{O}(T/(C^{1-2\alpha})+n^{2+2\alpha}/C^{2-2\alpha})$ time slots; and (c) the cost of each node is $\tilde{O}(\sqrt{T/(C^{1-2\alpha})}+n^{2+2\alpha}/C^{2-2\alpha})$.
\end{theorem}

Notice, when $C\leq n/2$, the $O(T/C)$ and $\tilde{O}(T/(C^{1-2\alpha}))$ term in the runtime of \MultiCastC and \MultiCastAdvC are optimal or near optimal, as Eve can jam all $C$ channels for at least $T/C$ slots. On the other hand, the $\tilde{O}(\sqrt{T/n})$ cost in \MultiCastC matches the currently best known 1-to-$n$ broadcast algorithm~\cite{gilbert14}, and the $\tilde{O}(\sqrt{T/(C^{1-2\alpha})}$ cost in \MultiCastAdvC remains to be very competitive. In short, our algorithms prove that having multiple channels indeed allows faster message dissemination, without sacrificing resource competitiveness. Even better, the more channels we have, the faster we can be!

Lastly, we note that it is not easy to covert a single-channel algorithm into a $C$-channels algorithm, with a $\Theta(C)$ speedup in time complexity. Specifically, the simple idea of ``grouping'' $C$ slots of the single-channel algorithm into one slot of the $C$-channels algorithm could fail: a node can broadcast or listen in multiple slots within a time window of $C$ slots in the single-channel setting, but usually it cannot access multiple channels simultaneously in one slot in a $C$-channels network.

\section{Future Work}\label{sec-future-work}

There are several interesting directions worth further exploration, and the most natural one is allowing Eve to be adaptive. We suspect \MultiCast and \MultiCastAdv can handle such more powerful adversary with few (or even no) modifications, while preserving time complexity and resource competitiveness. The main challenge, however, is to develop proper techniques for algorithm analysis. Another intriguing question is lower bounds. When $C=O(n)$, our algorithms (almost) meet the trivial $\Omega(T/C)$ time complexity lower bound, but what if $C=\omega(n)$? Besides, will having multiple channels or allowing Eve to be oblivious change the current lower bounds (see, e.g., \cite{gilbert14}) on energy consumption?


\clearpage
\bibliographystyle{ACM-Reference-Format}
\bibliography{./paper}

\clearpage
\appendix
\section*{Appendix}

\section{Omitted Proofs}\label{sec-appendix-proof}

\begin{proof}[\underline{Proof of Claim \ref{claim-multicastcore-epidemic-bcst-1}}]
Let $p=1/64$ denote the broadcast/listen probability of a node, and $c=n/2$ denote the number of channels. Fix a slot in $\mathcal{R}_1$ in which $t\leq n/2$, let $F$ be a random variable denoting the number of good channels. Let $F_i$ be an indicator random variable taking value one iff the $i$\textsuperscript{th} channel is good, let $U_i$ be an indicator random variable taking value one iff the $i$\textsuperscript{th} channel is not jammed by Eve, and let $V_i$ be an indicator random variable taking value one iff exactly one informed node broadcasts on channel $i$. Here, $1\leq i\leq c$. Since nodes take action independently, and since Eve does not know the random bits generated by nodes within current slot, we know $\Pr(F_i=1)=\Pr(U_i=1)\cdot\Pr(V_i=1)=\Pr(U_i=1)\cdot\left(t\cdot(p/c)\cdot(1-p/c)^{t-1}\right)\geq \Pr(U_i=1)\cdot(pt/c)\cdot\exp\left(-2\cdot(p/c)\cdot(t-1)\right)\geq \Pr(U_i=1)\cdot(pt/c)\cdot e^{-2p}$. By linearity of expectation, we know $\mathbb{E}[F]=\sum_{i=1}^{c}\mathbb{E}[F_i]=\sum_{i=1}^{c}\Pr(F_i=1)\geq (pt/c)\cdot e^{-2p}\cdot\sum_{i=1}^{c}\Pr(U_i=1)=(pt/c)\cdot e^{-2p}\cdot U$. Here, $U=\sum_{i=1}^{c}U_i$ denotes the number of channels not jammed by Eve in this slot. But we already know $U\geq c/10$, thus $\mathbb{E}[F]\geq (pt/10)\cdot e^{-2p}>t/1000$.

Our next step is to show the concentration of $F$ via the method of bounded differences (see, e.g., Chapter 5 of \cite{dubhashi09} for more details). Let $X_j$ be a random variable denoting the channel chosen by the $j$\textsuperscript{th} informed node, where $1\leq j\leq t$. Clearly, $F$ is a function of $(X_1,\cdots,X_t)$. I.e., $F=f(X_1,\cdots,X_t)$. Notice that $\{X_1,\cdots,X_t\}$ is a set of independent random variables. Moreover, the function $f$ satisfies the Lipschitz condition with constant two for each coordinate, as changing the channel assignment of any node changes the number of good channels by at most two. Thus, by the method of bounded differences (specifically, Corollary 5.2 of \cite{dubhashi09}), $\Pr(F<\mathbb{E}[F]-\mathbb{E}[F]/2)\leq\exp\left(-2\cdot(\mathbb{E}[F]/2)^2/(4t)\right)=\exp\left(-p^2t/(800e^{4p})\right)\leq\exp\left(-p^2/(800e^{4p})\right)<1-1/10^7$. I.e., with probability at least $1/10^7$, there are at least $t/2000$ good channels.
\end{proof}

\begin{proof}[\underline{Proof of Lemma \ref{lemma-multicast-epidemic-bcst}}]
If an iteration satisfies conditions (a) and (b), then in at least $2bi\cdot 4^i\cdot\lg^2{n}$ slots, at most $9n/20$ channels are jammed by Eve, where $b$ is some sufficiently large constant. Let $\mathcal{R}_1$ denote the collection of the first half of these $2bi\cdot 4^i\cdot\lg^2{n}$ slots, and let $\mathcal{R}_2$ denote the second half. We further divide $\mathcal{R}_1$ into $\lg{n}$ segments, each containing $bi\cdot 4^i\cdot\lg{n}$ slots. Throughout this proof, let $p=1/2^i$ denote the broadcasting/listening probability of a node, and $c=n/2$ denote the number of channels.

We first focus on slots in $\mathcal{R}_1$, and show that at least $n/2$ nodes will be informed by the end of $\mathcal{R}_1$. To see this, consider some segment $\mathcal{S}$ in $\mathcal{R}_1$, consider some node $u$ that is informed at the beginning of $\mathcal{S}$. Consider a slot in $\mathcal{S}$, let $t\in[1,n/2]$ denote the number of informed nodes at the beginning of this slot. Now, if $u$ wants to inform a previously uninformed node in this slot, the following conditions must hold: (a) $u$ decides to broadcast; (b) all other $t-1$ informed nodes do not broadcast on the channel chosen by $u$; (c) at least one uninformed node listens on the channel chosen by $u$; and (d) the channel chosen by $u$ is not jammed by Eve. Therefore, the probability that $u$ informs a previously uninformed node in this slot is at least $p\cdot(1-p/c)^{t-1}\cdot\left(1-(1-p/c)^{n-t}\right)\cdot(1/10)\geq p\cdot(1-p/c)^{n/2}\cdot(1-(1-p/c)^{n/2})\cdot(1/10)\geq p\cdot e^{-np/c}\cdot(1-e^{-np/2c})\cdot(1/10)\geq p\cdot e^{-np/c}\cdot\left(1-(1-np/4c)\right)\cdot(1/10)=(np^2/40c)\cdot e^{-np/c}=p^2/(20e^{2p})$. As a result, by the end of $\mathcal{S}$, the probability that no previously uninformed node is informed by $u$ during $\mathcal{S}$ is at most $\left(1-p^2/(20e^{2p})\right)^{bi\cdot 4^i\cdot\lg{n}}\leq e^{-(b/20e)\cdot i\cdot\lg{n}}$. Take a union bound over all the $O(n)$ nodes that are informed at the beginning of $\mathcal{S}$, we know the number of informed nodes will at least double by the end of $\mathcal{S}$, with probability at least $1-n\cdot e^{-(b/20e)\cdot i\cdot\lg{n}}$. Since there are $\lg{n}$ segments in $\mathcal{R}_1$, take another union bound over these segments and we know the number of informed nodes will at least reach $n/2$ by the end of $\mathcal{R}_1$, with probability at least $1-\lg{n}\cdot n\cdot e^{-(b/20e)\cdot i\cdot\lg{n}}$. For sufficiently large $b$, this probability is at least $1-1/n^{6i}$.

Next, we turn our attention to the slots in $\mathcal{R}_2$. Assume since the beginning of $\mathcal{R}_2$, at least $n/2$ nodes are informed. Consider a slot in $\mathcal{R}_2$ and a node $u$ that is still uninformed at the beginning of this slot. Let $t\in[n/2,n)$ denote the number of informed nodes at the beginning of this slot. Notice, for $u$ to become informed in this slot, the following conditions must hold: (a) $u$ decides to listen; (b) one informed node $v$ broadcasts on the channel chosen by $u$; (c) all informed nodes beside $v$ do not broadcast on the channel chosen by $u$; and (d) the channel chosen by $u$ is not jammed by Eve. Therefore, the probability that $u$ will be informed in this slot is at least $p\cdot t\cdot (p/c)\cdot\left(1-p/c\right)^{t-1}\cdot(1/10)\geq p\cdot(n/2)\cdot(p/c)\cdot\left(1-p/c\right)^{n}\cdot(1/10)\geq p\cdot(n/2)\cdot(p/c)\cdot e^{-2np/c}\cdot(1/10)=(p^2/10)\cdot e^{-4p}$. As a result, by the end of $\mathcal{R}_2$, the probability that $u$ is still uninformed is at most $\left(1-(p^2/10)\cdot e^{-4p}\right)^{bi\cdot 4^i\cdot\lg^2{n}}\leq e^{-(b/10e^2)\cdot i\cdot \lg^2{n}}$. Take a union bound over the $O(n)$ nodes that are uninformed at the beginning of $\mathcal{R}_2$, we know all nodes will be informed by the end of $\mathcal{R}_2$, with probability at least $1-n\cdot e^{-(b/10e^2)\cdot i\cdot \lg^2{n}}$. For sufficiently large $b$, this probability is at least $1-1/n^{6i}$.
\end{proof}

\begin{proof}[\underline{Proof of Lemma \ref{lemma-multicast-halt-imply-inform}}]
Let $\mathcal{E}_1$ be event (a), and $\mathcal{E}_2$ be event (b). Let $\mathcal{E}$ be the event that for at least ten percent of all slots within current iteration, Eve jams at most ninety percent of all $n/2$ channels. We know $\Pr(\mathcal{E}_1\wedge\mathcal{E}_2)=\Pr(\mathcal{E}_1\wedge\mathcal{E}_2\wedge\mathcal{E})+\Pr(\mathcal{E}_1\wedge\mathcal{E}_2\wedge\overline{\mathcal{E}})$.

Due to the proof of Lemma \ref{lemma-multicast-epidemic-bcst}, we know $\Pr(\mathcal{E}_1\wedge\mathcal{E}_2\wedge\mathcal{E})\leq\Pr(\mathcal{E}_2\wedge\mathcal{E})\leq\Pr(\mathcal{E}_2~|~\mathcal{E})\leq 2/n^{6i}$.

On the other hand, we bound $\Pr(\mathcal{E}_1\wedge\mathcal{E}_2\wedge\overline{\mathcal{E}})$ by $\Pr(\mathcal{E}_1~|~\overline{\mathcal{E}})$. If $\overline{\mathcal{E}}$ occurs, then Eve jams at least ninety percent of all $n/2$ channels for at least ninety percent of all slots within iteration $i$. Let $\mathcal{R}_1$ denote the set of slots in which Eve jams at least ninety percent of all $n/2$ channels. To make the number of noisy slots observed by $u$ as small as possible, without loss of generality, assume Eve jams exactly ninety percent of all $n/2$ channels for every slot in $\mathcal{R}_1$. Let $X_j$ be an indicator random variable taking value one iff $u$ listens on a channel jammed by Eve in the $j$\textsuperscript{th} slot in $\mathcal{R}_1$, where $1\leq j\leq|\mathcal{R}_1|$. Let $X=\sum_{j=1}^{|\mathcal{R}_1|}X_j$. We know $\mathbb{E}[X_j]=\Pr(X_j=1)=(9/10)\cdot(1/2^i)$, thus $\mathbb{E}[X]=|\mathcal{R}_1|\cdot\mathbb{E}[X_i]\geq (9R_i/10)\cdot\mathbb{E}[X_i]\geq 0.81R_i/2^i$. Notice that $X_1,X_2,\cdots,X_{|\mathcal{R}_1|}$ is a set of independent random variables, as nodes' channel choices and listening probabilities are not affected by Eve within an iteration, and we have assumed Eve jams exactly ninety percent of all $n/2$ channels for every slot in $\mathcal{R}_1$. Therefore, by a Chernoff bound, we know $\Pr(X< 0.5R_i/2^i)$ is exponential small in $R_i/2^i$. Specifically, $\Pr(\mathcal{E}_1~|~\overline{\mathcal{E}})\leq 1/n^{6i}$.
\end{proof}

\begin{proof}[\underline{Proof of Lemma \ref{lemma-multicast-competitive-halt}}]
Let $\mathcal{R}_1$ be the set of slots in which Eve jams more than twenty percent of all $n/2$ channels, and let $\mathcal{R}_2$ be the remaining slots. Let $R_1=|\mathcal{R}_1|$ and $R_2=|\mathcal{R}_2|$, we know $R_1\leq 0.2R_i$, $R_2\geq 0.8R_i$. To make the number of noisy slots observed by $u$ as large as possible, without loss of generality, we assume: (1) Eve jams all channels for every slot in $\mathcal{R}_1$ and Eve jams exactly twenty percent of all channels for every slot in $\mathcal{R}_2$; (2) all nodes are active and informed since the beginning of this iteration; and (3) $R_1=0.2R_i$ and $R_2=0.8R_i$.

Let $X_j$ be an indicator random variable taking value one iff $u$ decides to listen in the $j$\textsuperscript{th} slot in $\mathcal{R}_1$, where $1\leq j\leq R_1$. We know $\mathbb{E}[X_j]=\Pr(X_j=1)=1/2^i$.

Let $Y_j$ be an indicator random variable taking value one iff $u$ hears noise in the $j$\textsuperscript{th} slot in $\mathcal{R}_2$, where $1\leq j\leq R_2$. Further define $J_j$ be an indicator random variable taking value one iff the channel chosen by $u$ in the $j$\textsuperscript{th} slot in $\mathcal{R}_2$ is jammed by Eve, and define $I_j$ be a random variable denoting the number of broadcasting informed nodes on the channel chosen by $u$ in the $j$\textsuperscript{th} slot in $\mathcal{R}_2$. We know $\Pr(Y_j=1)=(1/2^i)\cdot\left(\Pr(J_j=1)+\Pr(I_j\geq 2)\cdot\Pr(J_j=0)\right)=(1/2^i)\cdot\left(0.2+0.8\cdot\Pr(I_j\geq 2)\right)$. To upper bound $\Pr(I_j\geq 2)=1-\Pr(I_j=0)-\Pr(I_j=1)$, we lower bound $\Pr(I_j=0)$. In particular, $\Pr(I_j=0)=\left(1-(1/2^i)\cdot(2/n)\right)^{n-1}\geq \exp\left(-2\cdot(n-1)\cdot(1/2^i)\cdot(2/n)\right)\geq \exp\left(-2\cdot n\cdot(1/64)\cdot(2/n)\right)\geq e^{-1/16}$. Hence, $\Pr(I_j\geq 2)\leq 1-\Pr(I_j=0)\leq 1-e^{-1/16}<0.061$, implying $\Pr(Y_j=1)<(1/2^i)\cdot0.25$.

Let $Z=\sum_{j=1}^{R_1}X_j + \sum_{j=1}^{R_2}Y_j$, we know $\mathbb{E}[Z]< R_1/2^i+(R_2/2^i)\cdot 0.25=0.4R_i/2^i$. Notice that $\{X_1,X_2,\cdots,X_{R_1},Y_1,Y_2,\cdots,Y_{R_2}\}$ is a set of independent random variables. Hence, by a Chernoff bound, we know $\Pr(Z\geq 0.5R_i/2^i)$ is exponentially small in $R_i/2^i$. Since $Z$ is an upper bound on the number of noisy slots $u$ will observe, the lemma follows.
\end{proof}

\begin{proof}[\underline{Proof of Theorem \ref{thm-multicast}}]
Let $l$ be the last iteration in which Eve jams more than twenty percent of all the $n/2$ channels for more than twenty percent of all slots, we know the cost of Eve in iteration $l$ is at least $(0.2\cdot n/2)\cdot(0.2\cdot al\cdot 4^l\cdot \lg^2{n})=0.02anl\cdot 4^l\cdot \lg^2{n}$. This also implies $T\geq 0.02anl\cdot 4^l\cdot \lg^2{n}$.

We first analyze the energy consumption of each node.

According to our protocol, for a fixed node $u$, for a fixed iteration $i$, node $u$'s expected cost is at most $2R_i\cdot p_i=2ai\cdot 2^i\cdot\lg^2{n}$. Apply a Chernoff bound, we know node $u$'s cost will be at most $3ai\cdot 2^i\cdot\lg^2{n}$, with probability at least $1-e^{-\Theta(i\cdot 2^i\cdot\lg^2{n})}$. Therefore, from the start of execution to the end of iteration $l$, with probability at least $1-\sum_{i=6}^{l}{e^{-\Theta(i\cdot 2^i\cdot\lg^2{n})}}= 1-\sum_{i=6}^{l}{(e^{-\Theta(\lg^2{n})})^{i\cdot 2^i}}\geq 1-\sum_{i=6}^{l}{(e^{-\Theta(\lg^2{n})})^{i}}\geq 1-1/n^{\Theta(\lg{n})}$, the total cost of node $u$ will be at most $\sum_{i=6}^{l}{3ai\cdot 2^i\cdot\lg^2{n}}=3a\lg^2{n}\cdot\sum_{i=6}^{l}{i\cdot 2^i}=3a\lg^2{n}\cdot((l-1)\cdot 2^{l+1}-2^8)\leq 3a\lg^2{n}\cdot l\cdot 2^{l+1}=6al\cdot 2^l\cdot \lg^2{n}$. On the other hand, due to Lemma \ref{lemma-multicast-competitive-halt}, we know by the end of iteration $l+1$, with probability at least $1-1/n^{\Theta(l\cdot 2^l\cdot\lg{n})}\geq 1-1/n^{\Theta(\lg{n})}$, node $u$ will terminate (if it has not terminated yet). Furthermore, in iteration $l+1$, with probability at least $1-1/n^{\Theta((l+1)\cdot 2^{l+1}\cdot\lg{n})}\geq 1-1/n^{\Theta(\lg{n})}$, the cost of node $u$ will be at most $3a(l+1)\cdot 2^{l+1}\cdot\lg^2{n}=6a(l+1)\cdot 2^{l}\cdot\lg^2{n}\leq 7al\cdot 2^{l}\cdot\lg^2{n}$. By now, we know $u$'s total cost during entire execution is at most $13al\cdot 2^{l}\cdot\lg^2{n}$, with probability at least $1-1/n^{\Theta(\lg{n})}$.

Let $E=13al\cdot 2^{l}\cdot\lg^2{n}$ be the total cost of node $u$, we know $E^2=169\cdot a^2\cdot l^2\cdot 4^l\cdot \lg^4{n}=(169al\lg^2{n})\cdot al\cdot 4^l\cdot\lg^2{n}\leq (169al\lg^2{n})\cdot T/(0.02n)=O(l\cdot (T/n)\cdot\lg^2{n})$, where the inequality follows by $T\geq 0.02anl\cdot 4^l\cdot \lg^2{n}$. Moreover, since $T\geq 0.02anl\cdot 4^l\cdot \lg^2{n}$, we also know $T\geq 0.02a\cdot 4^l$, obtaining $l=O(\lg{T})$. Thus, $E^2=O(\lg{T}\cdot(T/n)\cdot\lg^2{n})$, implying $E=O(\sqrt{T/n}\cdot\sqrt{\lg{T}}\cdot\lg{n})$. Take a union bound over all $n$ nodes, we know w.h.p., the total cost of each node is $O(\sqrt{T/n}\cdot\sqrt{\lg{T}}\cdot\lg{n})$.

Next, we analyze how long each node will remain active.

Similar to the above analysis, from the start of execution to the end of iteration $l$, the total time consumption is $\sum_{i=6}^{l}{ai\cdot 4^i\cdot \lg^2{n}}=a\lg^2{n}\cdot\sum_{i=6}^{l}i\cdot 4^i=a\lg^2{n}\cdot((l-1/3)\cdot4^{l+1}-(14/3)\cdot 4^6)/3\leq a\lg^2{n}\cdot(l\cdot4^{l+1})/3\leq 2a\lg^2{n}\cdot l\cdot 4^l=O(T/n)$. On the other hand, the length of iteration $l+1$ is $a(l+1)\cdot 4^{l+1}\cdot\lg^2{n}\leq 8al\cdot 4^l\cdot\lg^2{n}=O(T/n)$. Due to Lemma \ref{lemma-multicast-competitive-halt}, we know all nodes will terminate by the end of iteration $l+1$, w.h.p. Therefore, all nodes will terminate within $O(T/n)$ slots, w.h.p.

Lastly, we show that each node must have been informed when it terminates:

\vspace{-3ex}
\begin{align*}
& \Pr(\textrm{some node halts while uninformed})\\
= & \sum_{u}\Pr{{\textrm{some node halts while uninformed}}\choose{\textrm{and the first node that halts is }u}}\\
\leq & n\cdot\Pr{{\textrm{some node halts while uninformed}}\choose{\textrm{and the first node that halts is }u}}\\
\leq & n\cdot\left(\sum_{i=6}^{\infty}\Pr{{\textrm{some node halts while uninformed and the}}\choose{\textrm{first node that halts is }u\textrm{ and it halts at iteration }i}}\right)\\
\leq & n\cdot\sum_{i=6}^{\infty}1/n^{5i}=1/n^{\Omega(1)}\\
\end{align*}
\vspace{-6ex}

\noindent where the last inequality is due to Lemma \ref{lemma-multicast-halt-imply-inform}.

Before concluding the proof, we note that it is easy to verify, when Eve is not present (i.e., $T=0$) or $T=o(n)$, w.h.p.\ all nodes will be able to terminate by the end of the first iteration. In such scenario, the time and energy cost for each node is $O(\lg^2{n})$.
\end{proof}

\begin{proof}[\underline{Proof of Lemma \ref{lemma-multicastadv-ij-range-1}}]
First, consider a fixed $(i,j)$-phase in which $i\leq \lg(n/(d\lg{n}))$, where $d$ is some sufficiently large constant. Assume all nodes are active during this phase. To make $N_m$ as large as possible, without loss of generality, assume Eve does no jamming in this phase. Hence, in a slot in step two, the probability that a fixed node $u$ hears message $m$ is $p\cdot t\cdot (p/2^j)\cdot(1-p/2^j)^{n-2}$, where $t$ denotes the number of nodes that already know message $m$ at the beginning of step two. Since $t\leq n$, we know this probability is at most $pn\cdot(p/2^j)\cdot(1-p/2^j)^{n/2}\leq n\cdot (1-p/2^j)^{n/2}\leq n\cdot\exp(-pn/2^{j+1})=n\cdot\exp(-n/(4\cdot 2^{\alpha i+(1-\alpha)j}))\leq n\cdot\exp(-n/(4\cdot 2^{\alpha i+(1-\alpha)(i-1)}))= n\cdot\exp(-n/(4\cdot 2^{i-(1-\alpha)}))\leq n\cdot\exp\left(-n/\left(4\cdot 2^{i}\right)\right)\leq n\cdot\exp(-(d/4)\cdot\lg{n})=1/n^{\Omega(1)}$. Notice by the end of epoch $\lg(n/(d\lg{n}))$, the total number of slots in step two is bounded by $[\lg(n/(d\lg{n}))]\cdot [\lg(n/(d\lg{n}))]\cdot [b\cdot 2^{2\alpha\cdot \lg(n/(d\lg{n}))}\cdot  (\lg(n/(d\lg{n})))^3]=O(n\cdot \lg^5{n})$. Take a union bound over these $O(n\cdot \lg^5{n})$ slots, and another union bound over the $n$ nodes, we know during step two of all $(i,j)$-phases in which $i\leq  \lg(n/(d\lg{n}))$, no node will ever hear $m$, w.h.p. Thus no node will become \texttt{helper} during the first $\lg(n/(d\lg{n}))$ epochs, w.h.p.

Next, we focus on epochs $\lg(n/(d\lg{n})) \leq i\leq \lg{n}$. Consider a fixed $(i,j)$-phase in these epochs. Assume all nodes are active during this phase. To make $N_m$ as large as possible, without loss of generality, assume Eve does no jamming in this phase. Hence, in a slot in step two, the probability that a fixed node $u$ hears message $m$ is $p\cdot t\cdot (p/2^j)\cdot(1-p/2^j)^{n-2}$, where $t$ denotes the number of nodes that know $m$ at the beginning of step two. Since $t\leq n$, we know this probability is at most $pn\cdot(p/2^j)\cdot(1-p/2^j)^{0.98n}\leq p^2\cdot(n/2^j)\cdot\exp\left(-p\cdot(n/2^j)\cdot 0.98\right)\leq p^2\cdot(n/2^j)\cdot\exp\left(-0.49\cdot(n/2^j)^{1-\alpha}\right)$. Thus, $\mathbb{E}[N_m]\leq Rp^2\cdot(n/2^j)\cdot\exp\left(-0.49\cdot(n/2^j)^{1-\alpha}\right)$. Define function $f(x,\alpha)=x\cdot\exp\left(-0.49\cdot x^{1-\alpha}\right)$, it is not hard to verify: for all $x\geq 2$ and $0<\alpha<1/4$, $f(x,\alpha)\leq 1.1$. Thus, $\mathbb{E}[N_m]\leq 1.1 Rp^2$. Since $Rp^2=\Omega(\lg^3{n})$ for each $(i,j)$-phase when $\lg(n/(d\lg{n})) \leq i\leq \lg{n}$, and since whether $u$ will hear $m$ are independent among slots in step two (recall we have assumed Eve does no jamming), apply a standard Chernoff bound and we know, $N_m<1.5Rp^2$, w.h.p. Take a union bound over all $n$ nodes and another union bound over all $O(\lg^2{n})$ phases during epochs $\lg(n/(d\lg{n})) \leq i\leq \lg{n}$, we know no node will become \texttt{helper} during these epochs, w.h.p.
\end{proof}

\begin{proof}[\underline{Proof of Lemma \ref{lemma-multicastadv-ij-range-2}}]
Let $n'$ and $t$ denote the number of active nodes and the number of active nodes that know message $m$ at the beginning of step two of this phase, respectively. To make $N_m$ as large as possible, without loss of generality, assume Eve does no jamming in this phase. Hence, in a slot in step two, the probability that node $u$ hears message $m$ is at most $p\cdot t\cdot (p/2^j)\cdot(1-p/2^j)^{n'-2}$. Since $t\leq n$ and $2^j\cdot(1-p(i,j)/2^j)^{-(n'-2)}>n$, we know $t\leq 2^j\cdot\left(1-p/2^j\right)^{-(n'-2)}$. Thus, in a slot in step two, the probability that node $u$ hears message $m$ is at most $p^2$. As a result, $\mathbb{E}[N_m]\leq Rp^2$. Since $Rp^2=\Omega(i^3)$ for each $(i,j)$-phase, and since whether $u$ will hear $m$ are independent among slots in step two (recall we have assumed Eve does no jamming), apply a standard Chernoff bound and we know, the probability that $N_m\geq 1.5Rp^2$ is at most $\exp(-\Omega(i^3))$. This probability is at most $n^{-\Theta(i^2)}$ when $i>\lg{n}$, thus the lemma is proved.
\end{proof}

\begin{proof}[\underline{Proof of Lemma \ref{lemma-multicastadv-ij-range-3}}]
Let $q_k$ be a random variable denoting the fraction of channels that are not jammed by Eve in the $k$\textsuperscript{th} slot in step two of phase $(i,j)$, where $1 \leq k \leq R$. Let $Q=\sum_{k=1}^{R}q_k$. We will prove if $j<\lg n-1$, then $N'_m\leq 2.2Rp^2$ and $N_s\geq 0.9Rp$ cannot happen simultaneously, which means $u$ cannot become \texttt{helper} in phase $(i,j)$, regardless of $Q$'s value. Now assume $j< \lg n-1$.

By linearity of expectation, we have $\mathbb{E}[N_s]=Q\cdot p \cdot (1-p/2^j)^{n-1}$ and $\mathbb{E}[N'_m]=Q \cdot p \cdot n \cdot (p/2^j) \cdot (1-p/2^j)^{n-2}$. Let $\mathcal{E}_1$ be the event that $N'_m\leq 2.2Rp^2$, $\mathcal{E}_2$ be the event that $N_s\geq 0.9Rp$, and $\mathcal{E}$ be the event that $(1-p/2^j)^{n-1}\geq 0.8R/Q$. We know $\Pr(\mathcal{E}_1\wedge\mathcal{E}_2)=\Pr(\mathcal{E}_1\wedge\mathcal{E}_2\wedge\mathcal{E})+\Pr(\mathcal{E}_1\wedge\mathcal{E}_2\wedge\overline{\mathcal{E}})$.

If $\mathcal{E}$ occurs, then $\mathbb{E}[N'_m]=(1-p/2^j)^{n-2} \cdot (n/2^j) \cdot Qp^2\geq (1-p/2^j)^{2(n-1)}\cdot 4\cdot Q/R\cdot Qp^2\geq 4\cdot 0.8 \cdot 0.8 \cdot Rp^2=2.56Rp^2$. By a Chernoff bound, $\Pr(\mathcal{E}_1\wedge\mathcal{E}_2\wedge\mathcal{E})\leq\Pr(\mathcal{E}_1\wedge\mathcal{E})\leq\Pr(\mathcal{E}_1~|~\mathcal{E})= \Pr(N'_m\leq 2.2Rp^2~|~\mathcal{E})\leq e^{-\Theta(Rp^2)}\leq n^{-\Theta(i^2)}$. 

On the other hand, we bound $\Pr(\mathcal{E}_1\wedge\mathcal{E}_2\wedge\overline{\mathcal{E}})$ by $\Pr(\mathcal{E}_2~|~\overline{\mathcal{E}})$. If $\overline{\mathcal{E}}$ occurs, then $\mathbb{E}[N_s]\leq 0.8Rp$. By a Chernoff bound, $\Pr(\mathcal{E}_2~|~\overline{\mathcal{E}})= \Pr(N_s\geq 0.9Rp~|~\overline{\mathcal{E}}) \leq e^{-\Theta(Rp)}\leq n^{-\Theta(2^{\alpha i}\cdot i^2)}$.
\end{proof}

\begin{proof}[\underline{Proof of Lemma \ref{lemma-multicastadv-helper-imply-inform}}]
Let $\mathcal{E}_1$ be the event that some node has $N_m\geq 1.5Rp^2$ by the end of step two, and $\mathcal{E}_2$ be the event that some node still does not know the message $m$ by the end of step two. We consider two scenarios depending on whether the jamming from Eve is strong or not. More specifically, define $q_k$ be a random variable denoting the fraction of channels that are not jammed by Eve in the $k$\textsuperscript{th} slot in step two of phase $(i,j)$, where $1 \leq k \leq R$. Let $Q=\sum_{k=1}^{R}q_k$. Let $n'$ and $t$ denote the number of active nodes and the number of active nodes that know message $m$ at the beginning of step two of this phase, respectively. Define $\mathcal{E}_3$ be the event that $Q\leq R\cdot  2^j(1-p/2^j)^{-(n'-2)}/t$, which indicates the jamming from Eve is strong. We know $\Pr(\mathcal{E}_1\mathcal{E}_2)=\Pr(\mathcal{E}_1\mathcal{E}_2\mathcal{E}_3)+\Pr(\mathcal{E}_1\mathcal{E}_2\overline{\mathcal{E}_3})$.

We first focus on $\Pr(\mathcal{E}_1\mathcal{E}_2\mathcal{E}_3)$. Since $\Pr(\mathcal{E}_1\mathcal{E}_2\mathcal{E}_3)\leq \Pr(\mathcal{E}_1\mathcal{E}_3)\leq \Pr(\mathcal{E}_1~|~\mathcal{E}_3)$, we bound $\Pr(\mathcal{E}_1~|~\mathcal{E}_3)$. Fix a node $u$. Let $t$ denote the number of nodes that already know $m$ at the beginning of step two. Thus, $\mathbb{E}[N_m]=\sum_{k=1}^{R} q_k \cdot p t (p/2^j)(1-p/2^j)^{n'-2}=Q \cdot pt \cdot (p/2^j)(1-p/2^j)^{n'-2}\leq Rp^2=\Theta(i^3)$. Since $i>\lg{n}$, by a Chernoff bound, the probability that $N_m$ reaches $1.5Rp^2$ is at most $n^{-\Theta(i^2)}$. Take a union bound over all $n$ nodes, we know $\Pr(\mathcal{E}_1~|~\mathcal{E}_3)\leq n^{-\Theta(i^2)}$.

Then, we consider $\Pr(\mathcal{E}_1\mathcal{E}_2\overline{\mathcal{E}_3})$. Since  $\Pr(\mathcal{E}_1\mathcal{E}_2\overline{\mathcal{E}_3})\leq \Pr(\mathcal{E}_2\overline{\mathcal{E}_3})\leq \Pr(\mathcal{E}_2~|~\overline{\mathcal{E}_3})$, we bound $\Pr(\mathcal{E}_2~|~\overline{\mathcal{E}_3})$. Fix a node $u$. When $Q>R\cdot  2^j(1-p/2^j)^{-(n'-2)}/t$, we have $\mathbb{E}[N_m]>Rp^2=\Theta(i^3)$. But $i>\lg{n}$, thus by a Chernoff bound, the probability that $N_m$ is zero is at most $n^{-\Theta(i^2)}$. Take a union bound over all $n$ nodes, we know $\Pr(\mathcal{E}_2~|~\overline{\mathcal{E}_3})\leq n^{-\Theta(i^2)}$. Thus the lemma is proved.
\end{proof}

\begin{proof}[\underline{Proof of Lemma \ref{lemma-multicastadv-halt-imply-helper}}]
Let $\mathcal{E}_1$ be the event that $u$ decides to halt by the end of step two, and $\mathcal{E}_2$ be the event that some node still has not obtained \texttt{helper} status by the end of step two. Let $\mathcal{E}_3$ be the event that ``during step two, for at least 0.02 fraction of slots, Eve jams at least 0.02 fraction of all channels''. We know $\Pr(\mathcal{E}_1\mathcal{E}_2)=\Pr(\mathcal{E}_1\mathcal{E}_2\mathcal{E}_3)+\Pr(\mathcal{E}_1\mathcal{E}_2\overline{\mathcal{E}_3})$.

Bounding $\Pr(\mathcal{E}_1\mathcal{E}_2\mathcal{E}_3)$ is easy. Since $\Pr(\mathcal{E}_1\mathcal{E}_2\mathcal{E}_3)\leq \Pr(\mathcal{E}_1\mathcal{E}_3)\leq \Pr(\mathcal{E}_1~|~\mathcal{E}_3)$, we bound $\Pr(\mathcal{E}_1~|~\mathcal{E}_3)$. Specifically, we lower bound the number of noisy slots observed by $u$ during step two. To this end, without loss of generality, assume during step two, Eve jams 0.02 fraction of all channels in 0.02 fraction of all slots, and does no jamming in other slots. We focus on the slots that Eve does jam, call this set of slots $\mathcal{R}_1$. In each slot in $\mathcal{R}_1$, the probability that $u$ listens on a channel jammed by Eve is 0.02. Thus, in expectation, among slots in $\mathcal{R}_1$, $u$ will choose to listen on a channel jammed by Eve for $|\mathcal{R}_1|\cdot p\cdot 0.02=0.02^2Rp=Rp/2500$ times. Since among these slots whether $u$ will choose to listen on a channel jammed by Eve are independent, and since $Rp=\Omega(i^3)$ and $i>\lg{n}$, apply a Chernoff bound and we know, with probability at most $n^{-\Theta(i^2)}$, node $u$ will hear less than $Rp/3000$ noisy slots during step two.

Bounding $\Pr(\mathcal{E}_1\mathcal{E}_2\overline{\mathcal{E}_3})$ is more challenging. Since $\Pr(\mathcal{E}_1\mathcal{E}_2\overline{\mathcal{E}_3})\leq \Pr(\mathcal{E}_2~|~\mathcal{E}_1\overline{\mathcal{E}_3})$, we bound $\Pr(\mathcal{E}_2~|~\mathcal{E}_1\overline{\mathcal{E}_3})$.

Assume $u$ decides to halt by the end of phase $(\tilde{i},\tilde{j})$. Since $u$ obtained \texttt{helper} status in phase $(\hat{i},\hat{j})$, we know $\tilde{i}\geq \hat{i}+2/\alpha$ and $\tilde{j}=\hat{j}$.

We need to inspect phase $(\hat{i},\hat{j})$ more carefully, and we begin with two claims.

\begin{claim}\label{claim-multicastadv-helper-phase-relation-1}
Fix an $(i,j)$-phase where $i>\lg{n}$ and all nodes are active, fix a node $u$. The following two events happen simultaneously with probability at most $n^{-\Theta(i^2)}$: (a) $((n-1)/2^j)\cdot(1-p/2^j)^{n-2}<1.45$; and (b) $N_m\geq 1.5Rp^2$.
\end{claim}

\begin{proof}
Imagine besides the real execution $\Pi$ there is another execution $\Pi'$ in which nodes use same source of randomness, but Eve does no jamming in step two of phase $(i,j)$. In phase $(i,j)$ in $\Pi'$, $\mathbb{E}_{\Pi'}[N_m]=R\cdot p\cdot t\cdot(p/2^j)\cdot(1-p/2^j)^{n-2}\leq Rp^2\cdot((n-1)/2^j)\cdot(1-p/2^j)^{n-2}$, where $t$ denotes the number of nodes that know $m$ (excluding $u$ if $u$ also knows $m$). If $((n-1)/2^j)\cdot(1-p/2^j)^{n-2}<1.45$, then $\mathbb{E}_{\Pi'}[N_m]<1.45Rp^2$. Since Eve does no jamming in phase $(i,j)$ in $\Pi'$, apply a Chernoff bound and we know, $\Pr_{\Pi'}(N_m\geq 1.5Rp^2)\leq e^{-\Theta(Rp^2)}\leq n^{-\Theta(i^2)}$. But we know $\Pr_{\Pi}(N_m\geq 1.5Rp^2)\leq \Pr_{\Pi'}(N_m\geq 1.5Rp^2)$, hence $\Pr_{\Pi}(N_m\geq 1.5Rp^2)\leq n^{-\Theta(i^2)}$.
\end{proof}

\begin{claim}\label{claim-multicastadv-helper-phase-relation-2}
Fix an $(i,j)$-phase where $i>\lg{n}$ and all nodes are active, fix a node $u$. The following two events happen simultaneously with probability at most $n^{-\Theta(i^2)}$: (a) $(1-p/2^j)^{n-1}<0.89$; and (b) $N_s\geq 0.9Rp$.
\end{claim}

\begin{proof}
Imagine besides the real execution $\Pi$ there is another execution $\Pi'$ in which nodes use same source of randomness, but Eve does no jamming in step two of phase $(i,j)$. In phase $(i,j)$ in $\Pi'$, $\mathbb{E}_{\Pi'}[N_s]=Rp\cdot(1-p/2^j)^{n-1}$. If $(1-p/2^j)^{n-1}<0.89$, then $\mathbb{E}_{\Pi'}[N_s]<0.89Rp$. Since Eve does no jamming in phase $(i,j)$ in $\Pi'$, apply a Chernoff bound and we know, $\Pr_{\Pi'}(N_s\geq 0.9Rp)\leq e^{-\Theta(Rp)}\leq n^{-\Theta(i^2)}$. But we know $\Pr_{\Pi}(N_s\geq 0.9Rp)\leq \Pr_{\Pi'}(N_s\geq 0.9Rp)$, hence $\Pr_{\Pi}(N_s\geq 0.9Rp)\leq n^{-\Theta(i^2)}$.
\end{proof}

The above claims suggest, in phase $(\hat{i},\hat{j})$, with probability at least $1-n^{-\Theta(\hat{i}^2)}$, we have $((n-1)/2^{\hat{j}})\cdot(1-p(\hat{i},\hat{j})/2^{\hat{j}})^{n-2}\geq 1.45$ and $(1-p(\hat{i},\hat{j})/2^{\hat{j}})^{n-1}\geq 0.89$. Since $\tilde{i}\geq \hat{i}+2/\alpha$ and $\tilde{j}=\hat{j}$, we also know $p(\tilde{i},\tilde{j})\leq p(\hat{i},\hat{j})/4$. Moreover, due to Lemma \ref{lemma-multicastadv-ij-range-2} and Lemma \ref{lemma-multicastadv-ij-range-3}, it must be the case that $\hat{j}=\lg{n}-1$, with probability at least $1-1/n^{\Theta(\hat{i}^2)}$. Assume indeed $((n-1)/2^{\hat{j}})\cdot(1-p(\hat{i},\hat{j})/2^{\hat{j}})^{n-2}\geq 1.45$,  $(1-p(\hat{i},\hat{j})/2^{\hat{j}})^{n-1}\geq 0.89$, $p(\tilde{i},\tilde{j})\leq p(\hat{i},\hat{j})/4$, and $\hat{j}=\lg{n}-1$.

By now, we can conclude $((n-1)/2^{\tilde{j}})\cdot(1-p(\tilde{i},\tilde{j})/2^{\tilde{j}})^{n-2}=((n-1)/2^{\hat{j}})\cdot(1-p(\tilde{i},\tilde{j})/2^{\hat{j}})^{n-2}\geq ((n-1)/2^{\hat{j}})\cdot\exp(-2(n-2)\cdot p(\tilde{i},\tilde{j})/2^{\hat{j}})\geq ((n-1)/2^{\hat{j}})\cdot\exp(-(n-2)\cdot p(\hat{i},\hat{j})/2^{\hat{j}+1})\geq ((n-1)/2^{\hat{j}})\cdot(1-p(\hat{i},\hat{j})/2^{\hat{j}})^{(n-2)/2}\geq ((n-1)/2^{\hat{j}})^{1/2}\cdot (1.45)^{1/2}\geq (2(n-1)/n)^{1/2}\cdot (1.45)^{1/2}\geq (2\cdot 7/8)^{1/2}\cdot (1.45)^{1/2}>1.59$, with probability at least $1-n^{-\Theta(\hat{i}^2)}$. Similarly, we can also conclude, $(1-p(\tilde{i},\tilde{j})/2^{\tilde{j}})^{n-1}\geq (1-p(\hat{i},\hat{j})/2^{\hat{j}})^{(n-1)/2}\geq (0.89)^{1/2}>0.94$, with probability at least $1-n^{-\Theta(\hat{i}^2)}$. Assume in phase $(\tilde{i},\tilde{j})$, indeed $((n-1)/2^{\tilde{j}})\cdot(1-p(\tilde{i},\tilde{j})/2^{\tilde{j}})^{n-2}\geq 1.59$ and $(1-p(\tilde{i},\tilde{j})/2^{\tilde{j}})^{n-1}\geq 0.94$.

Besides, by Lemma \ref{lemma-multicastadv-helper-imply-inform}, by the end of phase $(\hat{i},\hat{j})$, all nodes already know $m$, with probability at least $1-n^{-\Theta(\hat{i}^2)}$. Assume indeed all nodes know $m$ in phase $(\tilde{i},\tilde{j})$.

Now, fix a node $v$ that has not obtained \texttt{helper} status at the beginning of phase $(\tilde{i},\tilde{j})$. Let $N^v_m$ and $N^v_s$ denote the number of slots that $v$ hears message $m$ and silence in step two, respectively. Let $N^v_{m'}$ denote the number of slots that $v$ hears message $m$ or message $\pm$ in step two.

Since $\overline{\mathcal{E}_3}$ occurs, and since we intend to make $N^v_m$ and $N^v_s$ as small as possible, without loss of generality, assume Eve jams 0.02 fraction of all channels for 0.98 fraction of all slots during step two of phase $(\tilde{i},\tilde{j})$, and she jams all channels in other slots of step two. Let $\mathcal{R}_1$ denote the set of slots in which Eve only jams 0.02 fraction of all channels. We know $\mathbb{E}[N^v_m]=|\mathcal{R}_1|\cdot p(\tilde{i},\tilde{j})\cdot 0.98\cdot (n-1)\cdot(p(\tilde{i},\tilde{j})/2^{\tilde{j}})\cdot(1-p(\tilde{i},\tilde{j})/2^{\tilde{j}})^{n-2}=0.98^2\cdot R(\tilde{i},\tilde{j})\cdot(p(\tilde{i},\tilde{j}))^2\cdot ((n-1)/2^{\tilde{j}})\cdot(1-p(\tilde{i},\tilde{j})/2^{\tilde{j}})^{n-2}\geq 0.98^2\cdot R(\tilde{i},\tilde{j})\cdot(p(\tilde{i},\tilde{j}))^2\cdot 1.59\geq 1.52\cdot R(\tilde{i},\tilde{j})\cdot(p(\tilde{i},\tilde{j}))^2$. Similarly, we know $\mathbb{E}[N^v_s]=|\mathcal{R}_1|\cdot p(\tilde{i},\tilde{j})\cdot 0.98\cdot (1-p(\tilde{i},\tilde{j})/2^{\tilde{j}})^{n-1}=0.98^2\cdot R(\tilde{i},\tilde{j})\cdot p(\tilde{i},\tilde{j})\cdot (1-p(\tilde{i},\tilde{j})/2^{\tilde{j}})^{n-1}\geq 0.98^2\cdot R(\tilde{i},\tilde{j})\cdot p(\tilde{i},\tilde{j})\cdot 0.94\geq 0.902\cdot R(\tilde{i},\tilde{j})\cdot p(\tilde{i},\tilde{j})$. Since whether $v$ hear $m$ (or silence) are independent among slots in $\mathcal{R}_1$, by a Chernoff bound, we know the probability that $N^v_m<1.5\cdot R(\tilde{i},\tilde{j})\cdot(p(\tilde{i},\tilde{j}))^2$ or $N^v_s<0.9\cdot R(\tilde{i},\tilde{j})\cdot p(\tilde{i},\tilde{j})$ is exponentially small in $R(\tilde{i},\tilde{j})\cdot(p(\tilde{i},\tilde{j}))^2$.

On the other hand, to make $N^v_{m'}$ as large as possible, without loss of generality, assume Eve does no jamming during step two of phase $(\tilde{i},\tilde{j})$. Recall $\hat{j}=\lg{n}-1$, thus $\mathbb{E}[N^v_{m'}]\leq R(\tilde{i},\tilde{j})\cdot p(\tilde{i},\tilde{j})\cdot n\cdot (p(\tilde{i},\tilde{j})/2^{\tilde{j}}) \cdot (1-p(\tilde{i},\tilde{j})/2^{\tilde{j}})^{n-1}\leq n/2^{\hat{j}}\cdot R(\tilde{i},\tilde{j})\cdot (p(\tilde{i},\tilde{j}))^2=2R(\tilde{i},\tilde{j})\cdot (p(\tilde{i},\tilde{j}))^2$. Since whether $v$ hear $m$ (or $\pm$) are independent among slots in step two of phase $(\tilde{i},\tilde{j})$, by a Chernoff bound, we know the probability that $N^v_{m'}>2.2\cdot R(\tilde{i},\tilde{j})\cdot(p(\tilde{i},\tilde{j}))^2$ is exponentially small in $R(\tilde{i},\tilde{j})\cdot(p(\tilde{i},\tilde{j}))^2$.

At this point, take a union bound over all $O(n)$ nodes that have not obtained \texttt{helper} status at the beginning of phase $(\tilde{i},\tilde{j})$, we can finally conclude $\Pr(\mathcal{E}_2~|~\mathcal{E}_1\overline{\mathcal{E}_3})$ is at most $n^{-\Theta(\hat{i}^2)}$.
\end{proof}

\begin{proof}[\underline{Proof of Lemma \ref{lemma-multicastadv-fast-term-inform}}]
We use the same strategy as in the proof of Lemma \ref{lemma-multicast-epidemic-bcst}. Specifically, since event $\mathcal{E}^{\geq 0.1}_{Step1}(\leq 0.9)$ occurs, we know during step one of phase $j=\lg{n}-1$, there must exist $2d\cdot 2^{2\alpha i-2\alpha j}\cdot i^3$ slots in which Eve jams at most 0.9 fraction of all used channels. Here, $d$ is some sufficiently large constant. Let $\mathcal{R}_1$ denote the set of the first half of these $2d\cdot 2^{2\alpha i-2\alpha j}\cdot i^3$ slots, and let $\mathcal{R}_2$ denote the second half. We further divide $\mathcal{R}_1$ into $i>\lg{n}$ segments, each containing $d\cdot 2^{2\alpha i-2\alpha j}\cdot i^2$ slots. Throughout this proof, let $c=n/2$ denote the number of used channels.

We first focus on the slots in $\mathcal{R}_1$, and show that at least $n/2$ nodes will be \texttt{informed} by the end of $\mathcal{R}_1$. To see this, consider some segment $\mathcal{S}$ in $\mathcal{R}_1$, consider some node $u$ that is \texttt{informed} at the beginning of $\mathcal{S}$. Consider a slot in $\mathcal{S}$, let $t\in[1,n/2]$ denote the number of \texttt{informed} nodes at the beginning of this slot. The probability that $u$ informs a previously \texttt{uninformed} node in this slot is at least $p\cdot(1-p/c)^{t-1}\cdot\left(1-(1-p/c)^{n-t}\right)\cdot(1/10)\geq p\cdot(1-p/c)^{n/2}\cdot(1-(1-p/c)^{n/2})\cdot(1/10)\geq p\cdot e^{-np/c}\cdot(1-e^{-np/2c})\cdot(1/10)\geq p\cdot e^{-np/c}\cdot\left(1-(1-np/4c)\right)\cdot(1/10)=(np^2/40c)\cdot e^{-np/c}=p^2/(20e^{2p})$. As a result, by the end of $\mathcal{S}$, the probability that no previously \texttt{uninformed} node is informed by $u$ during $\mathcal{S}$ is at most $(1-p^2/(20e^{2p}))^{d\cdot 2^{2\alpha i-2\alpha j}\cdot i^2}\leq e^{-(d/80e^2)\cdot i\cdot\lg{n}}$. Take a union bound over all the $O(n)$ nodes that are \texttt{informed} at the beginning of $\mathcal{S}$, we know the number of \texttt{informed} nodes will at least double by the end of $\mathcal{S}$, with probability at least $1-n\cdot e^{-(d/80e^2)\cdot i\cdot\lg{n}}$. Since there are $i>\lg{n}$ segments in $\mathcal{R}_1$, take another union bound over these segments and we know the number of \texttt{informed} nodes will reach $n/2$ by the end of $\mathcal{R}_1$, with probability at least $1-i\cdot n\cdot e^{-(d/80e^2)\cdot i\cdot\lg{n}}$. For sufficiently large $d$, this probability is at least $1-n^{-\Theta(i)}$.

We now turn our attention to the slots in $\mathcal{R}_2$. Due to the above analysis, assume at the beginning of $\mathcal{R}_2$, at least $n/2$ nodes are \texttt{informed}. Consider a slot in $\mathcal{R}_2$ and a node $u$ that is still \texttt{uninformed} at the beginning of this slot. Let $t\in[n/2,n)$ denote the number of \texttt{informed} nodes at the beginning of this slot. The probability that $u$ will hear message $m$ in this slot is at least $p\cdot t\cdot (p/c)\cdot\left(1-p/c\right)^{t-1}\cdot(1/10)\geq p\cdot(n/2)\cdot(p/c)\cdot\left(1-p/c\right)^{n}\cdot(1/10)\geq p\cdot(n/2)\cdot(p/c)\cdot e^{-2np/c}\cdot(1/10)=(p^2/10)\cdot e^{-4p}$. As a result, by the end of $\mathcal{R}_2$, the probability that $u$ is still \texttt{uninformed} is at most $(1-(p^2/10)\cdot e^{-4p})^{d\cdot 2^{2\alpha i-2\alpha j}\cdot i^3}\leq e^{-(d/40e^4)\cdot i^2\cdot \lg{n}}$. Take a union bound over the at most $n/2$ nodes that are \texttt{uninformed} at the beginning of $\mathcal{R}_2$, we know all nodes will be \texttt{informed} by the end of $\mathcal{R}_2$, with probability at least $1-n\cdot e^{-(d/40e^4)\cdot i^2\cdot \lg{n}}$. For sufficiently large $d$, this probability is at least $1-n^{-\Theta(i^2)}$.
\end{proof}

\begin{proof}[\underline{Proof of Lemma \ref{lemma-multicastadv-fast-term-helper}}]
To prove the lemma, we only need to show by the end of phase $j=\lg{n}-1$, all nodes have \emph{at least} progressed to \texttt{helper} status, with sufficiently high probability.

Consider a node $u$ that is only \texttt{informed} at the beginning of step two of phase $\lg{n}-1$. Since $\mathcal{E}^{\geq 0.9999}_{Step2}(\leq 0.0001)$ occurs, $\mathcal{E}^{\geq 0.99}_{Step2}(\leq 0.01)$ must occur. Now, in order to obtain the smallest possible $N_m$ and $N_s$ for node $u$, without loss of generality, assume during step two, Eve jams 0.01 fraction of all channels in 0.99 fraction of slots, and she jams all channels in remaining slots. Let $\mathcal{R}_1$ denote the set of slots in which Eve does not jam all channels, we know $\mathbb{E}[N_m]=|\mathcal{R}_1|\cdot p\cdot 0.99\cdot (n-1)\cdot(p/2^j)\cdot(1-p/2^j)^{n-2}\geq 0.99R\cdot p\cdot 0.99\cdot 0.99n\cdot(p/2^j)\cdot e^{-2np/2^j}=0.99^3\cdot 2Rp^2\cdot e^{-4p}$, and $\mathbb{E}[N_s]=|\mathcal{R}_1|\cdot p\cdot 0.99\cdot (1-p/2^j)^{n-1}\geq 0.99R\cdot p\cdot 0.99\cdot e^{-2np/2^j}=0.99^2Rp\cdot e^{-4p}$. Since $i\geq \lg{n}+6/\alpha$, we know $p=2^{-\alpha(i-j)}/2\leq 2^{-\alpha(6/\alpha+1)}/2\leq 1/64$. Thus, $\mathbb{E}[N_m]\geq 1.8Rp^2$ and $\mathbb{E}[N_s]\geq 0.92Rp$. Since whether $u$ will hear message $m$ (or silence) are independent among slots in $\mathcal{R}_1$, by a Chernoff bound, we know the probability that $N_m$ is less than $1.5Rp^2$ or $N_s$ is less than $0.9Rp$ is at most $e^{-\Theta(Rp^2)}\leq n^{-\Theta(i^2)}$.

On the other hand, to obtain the largest possible $N'_m$ for node $u$, without loss of generality, assume during step two, Eve does no jamming. Thus, $\mathbb{E}[N'_m]=R\cdot p\cdot (n-1)\cdot(p/2^j)\cdot(1-p/2^j)^{n-2}\leq Rp^2\cdot(n-1)/2^j\leq 2Rp^2$. Since whether $u$ will hear $m$ (or $\pm$) are independent among slots, a Chernoff bound implies the probability that $N'_m$ is more than $2.2Rp^2$ is at most $e^{-\Theta(Rp^2)}\leq n^{-\Theta(i^2)}$.

Take a union bound over all the $O(n)$ \texttt{informed} nodes at the beginning of step two of phase $\lg{n}-1$, we know all nodes will at least progress to \texttt{helper} status by the end of phase $j=\lg{n}-1$, with probability at least $1-n^{-\Theta(i^2)}$.
\end{proof}

\begin{proof}[\underline{Proof of Lemma \ref{lemma-multicastadv-fast-term-halt}}]
By Claim \ref{claim-multicastadv-helper-phase-relation-2} in the proof of Lemma \ref{lemma-multicastadv-halt-imply-helper}, we know $(1-p(\hat{i},\hat{j})/2^{\hat{j}})^{n-1}\geq 0.89$, with probability at least $1-n^{-\Theta(\hat{i}^2)}$. Assume this is indeed the case. On the other hand, since $i\geq\hat{i}+11/\alpha$, we know in phase $j=\hat{j}$ in epoch $i$, $p(i,j)=2^{-\alpha(i-j)}/2=2^{-\alpha(i-\hat{j})}/2\leq 2^{-\alpha(\hat{i}-\hat{j}+11/\alpha)}/2=p(\hat{i},\hat{j})/2048$. Thus, $(1-p(i,j)/2^j)^{n-1}\geq \exp(-2(n-1)\cdot p(i,j)/2^j)\geq \exp(-2(n-1)\cdot p(\hat{i},\hat{j})/(2048\cdot 2^{\hat{j}}))\geq ((1-p(\hat{i},\hat{j})/2^{\hat{j}})^{n-1})^{1/1024}\geq 0.89^{1/1024}$.

Recall we assume during step two of phase $j=\hat{j}$ event $\mathcal{E}^{\geq 1-x_2}_{Step2}(\leq y_2)$ occurs, where $x_2=y_2=1/10^4$. Thus, to obtain the largest possible $N_n$ for node $u$, without loss of generality, assume during step two, Eve jams $y_2$ fraction of all channels in $1-x_2$ fraction of slots, and she jams all channels in remaining slots. Let $\mathcal{R}_1$ denote the set of slots in which Eve does not jam all channels, and $\mathcal{R}_2$ denote other slots. It is easy to verify, the number of noisy slots node $u$ observed during $\mathcal{R}_2$ is at most $1.01\cdot|\mathcal{R}_2|\cdot p(i,j)=1.01x_2\cdot R(i,j)\cdot p(i,j)$, with probability at least $1-n^{-\Theta(\hat{i}^2)}$. On the other hand, the expected number of silent slots observed by $u$ during $\mathcal{R}_1$ is at least $|\mathcal{R}_1|\cdot p(i,j)\cdot (1-y_2)\cdot (1-p(i,j)/2^j)^{n-1}=(1-x_2)(1-y_2)\cdot R(i,j)\cdot p(i,j)\cdot (1-p(i,j)/2^j)^{n-1}\geq 0.89^{1/1024}\cdot (1-x_2)(1-y_2)\cdot R(i,j)\cdot p(i,j)$. Thus, the expected number of noisy slots observed by $u$ during $\mathcal{R}_1$ is at most $(1-x_2)(1-0.89^{1/1024}\cdot (1-y_2))\cdot R(i,j)\cdot p(i,j)$. By a Chernoff bound, $u$ will hear at most $1.01\cdot(1-x_2)(1-0.89^{1/1024}\cdot (1-y_2))\cdot R(i,j)\cdot p(i,j)$ noisy slots during $\mathcal{R}_1$, with probability at least $1-n^{-\Theta(\hat{i}^2)}$. Since $1.01x_2+1.01\cdot(1-x_2)(1-0.89^{1/1024}\cdot (1-y_2))<1/3000$, the total number of noisy slots observed by $u$ during step two of phase $(i,j)$ will be less than $R(i,j)\cdot p(i,j)/3000$.
\end{proof}

\begin{proof}[\underline{Proof of Theorem \ref{thm-multicastadv}}]
To begin with, we show that when the first \texttt{helper} appears, all nodes must have already learned the message $m$, with high probability. More specifically:

\vspace{-2ex}
\begin{align*}
& \Pr{{\textrm{some node is \texttt{uninformed} when}}\choose{\textrm{some node becomes \texttt{helper}}}}\\
= & \sum_{u}\Pr{{\textrm{some node is \texttt{uninformed} when}}\choose{u\textrm{ becomes the first \texttt{helper}}}}\\
\leq & n\cdot\Pr{{\textrm{some node is \texttt{uninformed} when}}\choose{u\textrm{ becomes the first \texttt{helper}}}}\\
\leq & n\cdot\sum_{i=\lg{n}+1}^{\infty}\sum_{j=0}^{i-1}\Pr{{\textrm{some node is \texttt{uninformed} when}}\choose{u\textrm{ becomes the first \texttt{helper} in phase }(i,j)}}\\
& +n\cdot\Pr(u\textrm{ becomes the first \texttt{helper} in some epoch }\leq\lg{n})\\
\leq & \left(n\cdot\sum_{i=\lg{n}+1}^{\infty}\sum_{j=0}^{i-1}n^{-\Theta(i^2)}\right)+n^{-\Omega(1)}=n^{-\Omega(1)}
\end{align*}

\noindent and the last inequality is due to Lemma \ref{lemma-multicastadv-ij-range-1} and Lemma \ref{lemma-multicastadv-helper-imply-inform}.

We then show, when the first \texttt{halt} node appears, all nodes must have progressed to \texttt{helper} status, with high probability. See Figure \ref{fig-proof-multicastadv-1} for details, and the inequality in the second to last line in the figure is due to Lemma \ref{lemma-multicastadv-ij-range-1} and Lemma \ref{lemma-multicastadv-halt-imply-helper}.

\begin{figure*}[!t]
\begin{align*}
& \Pr(\textrm{some node is not \texttt{helper} when some node halts})\\
= & \sum_{v}\Pr(\textrm{some node is not \texttt{helper} when }v\textrm{ becomes the first \texttt{halt} node})\\
\leq & n\cdot\Pr(\textrm{some node is not \texttt{helper} when }v\textrm{ becomes the first \texttt{halt} node})\\
\leq & n\cdot\sum_{i=\lg{n}+1}^{\infty}\sum_{j=0}^{i-1}\Pr\left({\textrm{some node is not \texttt{helper} when}\choose{v\textrm{ becomes the first \texttt{halt} node}}}\vert{{v\textrm{ becomes \texttt{helper}}}\choose{\textrm{in phase }(i,j)}}\right)\cdot\Pr{{v\textrm{ becomes \texttt{helper}}}\choose{\textrm{in phase }(i,j)}}\\
& +n\cdot\Pr(v\textrm{ becomes \texttt{helper} in some epoch }\leq\lg{n})\\
\leq & \left(n\cdot\sum_{i=\lg{n}+1}^{\infty}\sum_{j=0}^{i-1}n^{-\Theta(i^2)}\cdot\Pr{{v\textrm{ becomes \texttt{helper}}}\choose{\textrm{in phase }(i,j)}}\right) + n^{-\Omega(1)}\\
\leq & \left(n^{-\Theta(\lg^2{n})}\cdot\sum_{i=\lg{n}+1}^{\infty}\sum_{j=0}^{i-1}\Pr{{v\textrm{ becomes \texttt{helper}}}\choose{\textrm{in phase }(i,j)}}\right) + n^{-\Omega(1)}\leq n^{-\Theta(\lg^2{n})}\cdot 1+n^{-\Omega(1)}=n^{-\Omega(1)}
\end{align*}
\vspace{-5ex}
\caption{}\label{fig-proof-multicastadv-1}
\end{figure*}

Our next claim is, for each node $u$, if it obtained \texttt{helper} status in phase $(\hat{i}_u,\hat{j}_u)$, then it must be the case that $\hat{i}_u>\lg{n}$ and $\hat{j}_u=\lg{n}-1$, with high probability.

\vspace{-2ex}
\begin{align*}
& \Pr(\textrm{some node becomes \texttt{helper} with }\hat{i}_u\leq\lg{n}\textrm{ or }\hat{j}_u\neq\lg{n}-1)\\
= & \sum_{u}\Pr(u\textrm{ becomes \texttt{helper} with }\hat{i}_u\leq\lg{n}\textrm{ or }\hat{j}_u\neq\lg{n}-1)\\
\leq & n\cdot\Pr(u\textrm{ becomes \texttt{helper} with }\hat{i}_u\leq\lg{n}\textrm{ or }\hat{j}_u\neq\lg{n}-1
)\\
= & n\cdot\Pr(u\textrm{ becomes \texttt{helper} with }\hat{i}_u>\lg{n}\textrm{ and }\hat{j}_u\neq\lg{n}-1)\\
& + n\cdot\Pr(u\textrm{ becomes \texttt{helper} with }\hat{i}_u\leq\lg{n})\\
= & n\cdot\Pr{{u\textrm{ becomes \texttt{helper} with }\hat{i}_u>\lg{n}\textrm{ and }\hat{j}_u\neq\lg{n}-1}\choose{\textrm{when all nodes are active}}}\\
& + n\cdot\Pr{{u\textrm{ becomes \texttt{helper} with }\hat{i}_u>\lg{n}\textrm{ and }\hat{j}_u\neq\lg{n}-1}\choose{\textrm{after some node halts}}}\\
& + n\cdot\Pr(u\textrm{ becomes \texttt{helper} with }\hat{i}_u\leq\lg{n})\\
\leq & n\cdot n^{-\Omega(1)} + n\cdot n^{-\Omega(1)} + n\cdot n^{-\Omega(1)} = n^{-\Omega(1)}
\end{align*}

\noindent and the last inequality is due to Lemma \ref{lemma-multicastadv-ij-range-1}, \ref{lemma-multicastadv-ij-range-2}, \ref{lemma-multicastadv-ij-range-3}, and Figure \ref{fig-proof-multicastadv-1}.

Based on above analysis, in the following proof, assume all nodes already know $m$ when the first \texttt{helper} appears, and all nodes have obtained \texttt{helper} status when the first \texttt{halt} node appears, and $\hat{i}_u>\lg{n}$ and $\hat{j}_u=\lg{n}-1$ for every node $u$.

Let $l$ denote the last epoch that is blocking while some node is still active at the beginning of epoch $l$, due to the above analysis, by the end of epoch $l$, we know the status of all nodes must belong to exactly one of the following four cases: (1) all nodes are active and either \texttt{uninformed} or \texttt{informed} (and there exists at least one \texttt{uninformed} node); (2) all nodes are active and either \texttt{informed} or \texttt{helper} (and there exists at least one \texttt{informed} node); (3) every node is either \texttt{helper} or has terminated (and there exists at least one \texttt{helper} node); or (4) all nodes have terminated.

Following analysis consider two scenarios: $l\geq\lg{n}$ and $l<\lg{n}$.

\textbf{The $l\geq\lg{n}$ scenario.} In this situation, if we are in case (3) by the end of epoch $l$: due to Lemma \ref{lemma-multicastadv-fast-term-halt}, w.h.p.\ all nodes will halt by the end of epoch $l+11/\alpha$. If we are in case (2) by the end of epoch $l$: due to Lemma \ref{lemma-multicastadv-fast-term-helper}, w.h.p.\ we must be in case (3) by the end of epoch $l+6/\alpha$; and due to Lemma \ref{lemma-multicastadv-fast-term-halt}, w.h.p.\ all nodes will halt by the end of epoch $l+17/\alpha$. Lastly, if we are in case (1) by the end of epoch $l$: due to Lemma \ref{lemma-multicastadv-fast-term-inform}, w.h.p.\ we must be in case (2) by the end of epoch $l+1$; due to Lemma \ref{lemma-multicastadv-fast-term-helper}, w.h.p.\ we must be in case (3) by the end of epoch $l+6/\alpha+1$; and due to Lemma \ref{lemma-multicastadv-fast-term-halt}, w.h.p.\ all nodes will halt by the end of epoch $l+17/\alpha+1$. By now, we know when $l\geq\lg{n}$, all nodes must have terminated by the end of epoch $l+17/\alpha+1\leq l+18/\alpha$, w.h.p.

To bound the runtime and cost of honest nodes, we need to bound $l$. Since epoch $l$ is blocking, we know $\mathcal{E}^{> x_1}_{Step1}(> y_1)$ or $\mathcal{E}^{> x_2}_{Step2}(> y_2)$ must have occurred during phase $\lg{n}-1$ of this epoch (notice $l\geq\lg{n}$ implies phase $\lg{n}-1$ must exist). Thus, the cost of Eve during epoch $l$ is at least $b\cdot x_2y_2\cdot 2^{\alpha(2l-2(\lg{n}-1))}\cdot l^3\cdot 2^{\lg{n-1}}=b\cdot x_2y_2\cdot 2^{2\alpha l}\cdot l^3\cdot 2^{(1-2\alpha)(\lg{n}-1)}$, implying $T\geq (x_2y_2/2)\cdot 2^{2\alpha l}\cdot l^3\cdot n^{(1-2\alpha)}\geq (x_2y_2/2)\cdot 2^{2\alpha l}\cdot n^{(1-2\alpha)}$. As a result, $l\leq (\lg(T/((x_2y_2/2)\cdot n^{(1-2\alpha)})))/(2\alpha)$.

At this point, we can conclude each node will halt within:

\vspace{-3ex}
\begin{align*}
& \sum_{i=1}^{l+18/\alpha}\sum_{j=0}^{i-1}\Theta(1)\cdot 2^{2\alpha i}\cdot 2^{-2\alpha j}\cdot i^3\leq \Theta\left(\frac{1}{1-2^{-2\alpha}}\right)\cdot \sum_{i=1}^{l+18/\alpha}2^{2\alpha i}\cdot i^3\\
\leq & \Theta\left(\frac{1}{1-2^{-2\alpha}}\right)\cdot\left(l+\frac{18}{\alpha}\right)^3\cdot \sum_{i=1}^{l+18/\alpha}2^{2\alpha i}\\
\leq & \Theta\left(\frac{1}{1-2^{-2\alpha}}\right)\cdot(2l)^3\cdot\frac{2^{2\alpha}\cdot 2^{2\alpha(l+18/\alpha)}}{2^{2\alpha}-1}\\
\leq & \Theta\left(\frac{1}{1-2^{-2\alpha}}\right)\cdot\left(\frac{2}{2\alpha}\lg{T}\right)^3\cdot\frac{2^{2\alpha+36}}{2^{2\alpha}-1}\cdot\frac{T}{(x_2y_2/2)\cdot n^{(1-2\alpha)}}\\
= & \Theta(1)\cdot \frac{T}{n^{(1-2\alpha)}}\cdot\lg^3{T}
\end{align*}

To bound each node's cost, first consider epochs one to $\lg{n}$. For each node, the total cost of it during these epochs cannot exceed the total length of these epochs. Thus, for each node, the total cost during the first $\lg{n}$ epochs is at most:

\vspace{-3ex}
\begin{align*}
& \sum_{i=1}^{\lg{n}}\sum_{j=0}^{i-1}\Theta(1)\cdot 2^{2\alpha i}\cdot 2^{-2\alpha j}\cdot i^3 \leq \Theta\left(\frac{1}{1-2^{-2\alpha}}\right)\cdot \sum_{i=1}^{\lg{n}}2^{2\alpha i}\cdot i^3\\
\leq & \Theta\left(\frac{1}{1-2^{-2\alpha}}\right)\cdot\lg^3{n}\cdot\sum_{i=1}^{\lg{n}}2^{2\alpha i}\\
\leq & \Theta\left(\frac{1}{1-2^{-2\alpha}}\right)\cdot\lg^3{n}\cdot\frac{2^{2\alpha}\cdot 2^{2\alpha\lg{n}}}{2^{2\alpha}-1}\\
= & \Theta\left(n^{2\alpha}\cdot\lg^3{n}\right)
\end{align*}

Starting from epoch $\lg{n}+1$, each node's cost within an epoch will be tightly concentrated around its expectation, with sufficiently high probability. Thus, with high probability, for every node, the total cost after epoch $\lg{n}$ is at most:

\vspace{-3ex}
\begin{align*}
& \sum_{i=1}^{l+18/\alpha}\sum_{j=0}^{i-1}\Theta(1)\cdot 2^{\alpha i}\cdot 2^{-\alpha j}\cdot i^3 \leq \Theta\left(\frac{1}{1-2^{-\alpha}}\right)\cdot \sum_{i=1}^{l+18/\alpha}2^{\alpha i}\cdot i^3\\
\leq & \Theta\left(\frac{1}{1-2^{-\alpha}}\right)\cdot\left(l+\frac{18}{\alpha}\right)^3\cdot \sum_{i=1}^{l+18/\alpha}2^{\alpha i}\\
\leq & \Theta\left(\frac{1}{1-2^{-\alpha}}\right)\cdot(2l)^3\cdot\frac{2^{\alpha}\cdot 2^{\alpha(l+18/\alpha)}}{2^{\alpha}-1}\\
\leq & \Theta\left(\frac{1}{1-2^{-\alpha}}\right)\cdot\left(\frac{2}{2\alpha}\lg{T}\right)^3\cdot\frac{2^{\alpha+18}}{2^{\alpha}-1}\cdot\left(\frac{T}{(x_2y_2/2)\cdot n^{(1-2\alpha)}}\right)^{1/2}\\
= & \Theta(1)\cdot\left(\frac{T}{n^{(1-2\alpha)}}\right)^{1/2}\cdot\lg^3{T}
\end{align*}

As a result, when $l\geq\lg{n}$, with high probability, the total cost of each node is $O((T/(n^{(1-2\alpha)}))^{1/2}\cdot\lg^3{T}+n^{2\alpha}\cdot\lg^3{n})$.

\textbf{The $l<\lg{n}$ scenario.} This situation is easier to analyze. By an analysis similar to the $l\geq\lg{n}$ scenario, we know all nodes will terminate by the end of epoch $\lg{n}+18/\alpha$, w.h.p. Therefore, the total runtime and cost for each node in this situation is at most:

\vspace{-2ex}
$$\sum_{i=1}^{\lg{n}+18/\alpha}\sum_{j=0}^{i-1}\Theta(1)\cdot 2^{2\alpha i}\cdot 2^{-2\alpha j}\cdot i^3\leq\Theta\left(n^{2\alpha}\cdot\lg^3{n}\right)$$

We conclude the proof by noting that the $l<\lg{n}$ scenario also includes the situation in which Eve is not present.
\end{proof}

\section{Omitted Figures}\label{sec-appendix-figure}

\begin{figure*}[!t]
\hrule
\vspace{1ex}\textbf{Pseudocode of \MultiCastAdv executed at node $u$:}\vspace{1ex}
\hrule
\begin{small}
\begin{algorithmic}[1]
\State $status\gets un$, $i\gets 1$.
\If {(node $u$ is the source node)} $status\gets in$. \EndIf
\Repeat
	\For {(each phase $j$ from $0$ to $i-1$)}
		\State $\texttt{MultiCastAdvPhase}(i,j)$. \Comment Execute phase $j$ of epoch $i$.
	\EndFor
	\State $i\gets i+1$.
\Until {($status==halt$)}
\end{algorithmic}
\end{small}
\hrule
\vspace{3ex}
\hrule
\vspace{1ex}\textbf{Pseudocode of }$\texttt{MultiCastAdv}(i,j)$ \textbf{executed at node $u$:}\vspace{1ex}
\hrule
\begin{small}
\begin{algorithmic}[1]
\State $R\gets b\cdot 2^{2\alpha(i-j)}\cdot i^3$, $p\gets 2^{-\alpha(i-j)}/2$. \Comment $b$ is some sufficiently large constant.
\Statex \textsc{Step I}: Message dissemination.
\For {(each slot from $1$ to $R$)}
	\State $ch\gets\texttt{rnd}(1,2^j)$, $coin\gets\texttt{rnd}(1,1/p)$.
	\If {($status==un$ \textbf{and} $coin==1$)}
		\State $feedback\gets\texttt{listen}(ch)$.
		\If {($feedback$ contains the message $m$)} $status\gets in$. \EndIf
	\ElsIf {($status\neq un$ \textbf{and} $coin==1$)}
		\State $\texttt{broadcast}(ch,m)$.
	\EndIf
\EndFor
\Statex \textsc{Step II}: Status adjustment and termination detection.
\State $N_m\gets 0$, $N'_m\gets 0$, $N_n\gets 0$, $N_s\gets 0$.
\For {(each slot from $1$ to $R$)}
	\State $ch\gets\texttt{rnd}(1,2^j)$, $coin\gets\texttt{rnd}(1,1/p)$.
	\If {($coin==1$)}
		\State $feedback\gets\texttt{listen}(ch)$.
		\If {($feedback$ contains the message $m$)} $N_m\gets N_m+1,N'_m\gets N'_m+1$.
		\ElsIf {($feedback$ contains beacon message $\pm$ )} $N'_m\gets N'_m+1$.
		\ElsIf {($feedback$ is noise)} $N_n\gets N_n+1$.
		\ElsIf {($feedback$ is silence)} $N_s\gets N_s+1$.
		\EndIf
	\ElsIf {($coin==2$)}
		\If {($status==un$)} $\texttt{broadcast}(ch,\pm)$.
		\Else\ $\texttt{broadcast}(ch,m)$.
		\EndIf
	\EndIf
\EndFor
\If {($status==un$ \textbf{and} $N_m\geq 1$)} $status\gets in$. \EndIf
\If {($status==in$ \textbf{and} $N_m\geq 1.5Rp^2$ \textbf{and} $N_s\geq 0.9Rp$ \textbf{and} $N'_m\leq 2.2Rp^2$)} $status\gets helper$, $\hat{i}_u\gets i$, $\hat{j}_u\gets j$. \EndIf
\If {($status==helper$ \textbf{and} $i-\hat{i}_u\geq 2/\alpha$ \textbf{and} $j==\hat{j}_u$ \textbf{and} $N_n\leq Rp/3000$)} $status\gets halt$. \EndIf
\end{algorithmic}
\end{small}
\hrule
\vspace{-2ex}
\caption{Pseudocode of the \MultiCastAdv algorithm.}\label{fig-alg-multicastadv}
\end{figure*}

\noindent Pseudocode of \MultiCastAdv: see Figure \ref{fig-alg-multicastadv}.

\begin{figure*}[!t]
\hrule
\vspace{1ex}\textbf{Pseudocode of \MultiCastC executed at node $u$:}\vspace{1ex}
\hrule
\begin{small}
\begin{algorithmic}[1]
\State $status\gets un$.
\If {(node $u$ is the source node)} $status\gets in$. \EndIf
\For {(each iteration $i\geq 6$)}
	\State $N_n\gets 0$.
	\For {(each round from $1$ to $R_i=ai\cdot 4^i\cdot \lg^2{n}$)}
		\State $ch\gets\texttt{rnd}(1,n/2)$, $coin\gets\texttt{rnd}(1,2^i)$.
		\For {(each slot $k$ from $1$ to $n/(2C)$)} \Comment Use $n/(2C)$ slots to simulate one slot in \MultiCast.
			\If {($k==\lfloor(ch-1)/C\rfloor+1$)} \Comment Operate in the $(\lfloor(ch-1)/C\rfloor+1)$\textsuperscript{th} slot within current round.
				\If {($coin==1$)}
					\State $feedback\gets\texttt{listen}(((ch-1)\mod C)+1)$.
					\If {($feedback$ is noise)}
						\State $N_n\gets N_n+1$.
					\ElsIf {($feedback$ contains the message $m$)}
						\State $status\gets in$.
					\EndIf
				\ElsIf {($coin==2$ \textbf{and} $status==in$)}
					\State $\texttt{broadcast}(((ch-1)\mod C)+1,m)$.
				\EndIf
			\EndIf
		\EndFor
	\EndFor
	\If {($N_n<R_i/2^{i+1}$)} \textbf{halt}. \EndIf
\EndFor
\end{algorithmic}
\end{small}
\hrule\vspace{-2ex}
\caption{Pseudocode of the \MultiCastC algorithm.}\label{fig-alg-multicastc}
\end{figure*}

\noindent Pseudocode of \MultiCastC: see Figure \ref{fig-alg-multicastc}.

\begin{figure*}[!t]
\hrule
\vspace{1ex}\textbf{Pseudocode of \MultiCastAdvC executed at node $u$:}\vspace{1ex}
\hrule
\begin{small}
\begin{algorithmic}[1]
\State $status\gets un$, $i\gets 1$.
\If {(node $u$ is the source node)} $status\gets in$. \EndIf
\Repeat
	\For {(each phase $j$ from $0$ to $\min\{i-1,\lg C\}$)}
		\State $\texttt{MultiCastAdvPhaseAlt}(i,j)$.
	\EndFor
	\State $i\gets i+1$.
\Until {($status==halt$)}
\end{algorithmic}
\end{small}
\hrule
\vspace{3ex}
\hrule
\vspace{1ex}\textbf{Pseudocode of }$\texttt{MultiCastAdvPhaseAlt}(i,j)$ \textbf{executed at node $u$:}\vspace{1ex}
\hrule
\begin{small}
\begin{algorithmic}[1]
\State $R\gets b\cdot 2^{2\alpha(i-j)}\cdot i^3$, $p\gets 2^{-\alpha(i-j)}/2$. \Comment $b$ is some sufficiently large constant.
\Statex \textsc{Step I}: Message dissemination.
\For {(each slot from $1$ to $R$)}
	\State $ch\gets\texttt{rnd}(1,2^j)$, $coin\gets\texttt{rnd}(1,1/p)$.
	\If {($status==un$ \textbf{and} $coin==1$)}
		\State $feedback\gets\texttt{listen}(ch)$.
		\If {($feedback$ contains the message $m$)} $status\gets in$. \EndIf
	\ElsIf {($status\neq un$ \textbf{and} $coin==1$)}
		\State $\texttt{broadcast}(ch,m)$.
	\EndIf
\EndFor
\Statex \textsc{Step II}: Status adjustment and termination detection.
\State $N_m\gets 0$, $N'_m\gets 0$, $N_n\gets 0$, $N_s\gets 0$.
\For {(each slot from $1$ to $R$)}
	\State $ch\gets\texttt{rnd}(1,2^j)$, $coin\gets\texttt{rnd}(1,1/p)$.
	\If {($coin==1$)}
		\State $feedback\gets\texttt{listen}(ch)$.
		\If {($feedback$ contains the message $m$)} $N_m\gets N_m+1,N'_m\gets N'_m+1$.
		\ElsIf {($feedback$ contains beacon message $\pm$ )} $N'_m\gets N'_m+1$.
		\ElsIf {($feedback$ is noise)} $N_n\gets N_n+1$.
		\ElsIf {($feedback$ is silence)} $N_s\gets N_s+1$.
		\EndIf
	\ElsIf {($coin==2$)}
		\If {($status==un$)} $\texttt{broadcast}(ch,\pm)$.
		\Else\ $\texttt{broadcast}(ch,m)$.
		\EndIf
	\EndIf
\EndFor
\If {($status==un$ \textbf{and} $N_m\geq 1$)} $status\gets in$. \EndIf
\If {($status==in$)}
	\If {($j==\lg{C}$ \textbf{and} $N_m\geq 1.5Rp^2$ \textbf{and} $N_s\geq 0.9Rp$)} $status\gets helper$, $\hat{i}_u\gets i$, $\hat{j}_u\gets j$.
	\ElsIf {($j<\lg{C}$ \textbf{and} $N_m\geq 1.5Rp^2$ \textbf{and} $N_s\geq 0.9Rp$ \textbf{and} $N'_m\leq 2.2Rp^2$)} $status\gets helper$, $\hat{i}_u\gets i$, $\hat{j}_u\gets j$.
	\EndIf
\EndIf
\If {($status==helper$ \textbf{and} $i-\hat{i}_u\geq 2/\alpha$ \textbf{and} $j==\hat{j}_u$ \textbf{and} $N_n\leq Rp/3000$)} $status\gets halt$. \EndIf
\end{algorithmic}
\end{small}
\hrule
\vspace{-2ex}
\caption{Pseudocode of the \MultiCastAdvC algorithm.}\label{fig-alg-multicastadvc}
\end{figure*}

\noindent Pseudocode of \MultiCastAdvC: see Figure \ref{fig-alg-multicastadvc}.

\section{Analysis of \MultiCastAdvC}\label{sec-appendix-multicastadvc}

As mentioned in the main body of the paper, here we focus on the scenario in which $C\leq n/2$.

\bigskip\noindent\textbf{Good $(i,j)$ phases.} It is easy to verify Lemma \ref{lemma-multicastadv-ij-range-1} and \ref{lemma-multicastadv-ij-range-2} still hold. However, now Lemma \ref{lemma-multicastadv-ij-range-2} is not that helpful, as \MultiCastAdvC does not contain phases in which $j\geq\lg{n}$ when $C\leq n/2$.

Lemma \ref{lemma-multicastadv-ij-range-3}, on the other hand, needs some care. Specifically, this lemma shows that nodes cannot obtain \texttt{helper} status when $j<\lg{n}-1$ and all nodes are active. Recall in \MultiCastAdv and \MultiCastAdvC, the pseudocode of an $(i,j)$-phase are identical when $j<\lg{C}$. Moreover, since $C\leq n/2$, we know $\lg{C}\leq\lg{n}-1$. Thus, when $j<\lg{C}$, Lemma \ref{lemma-multicastadv-ij-range-3} still holds. More specifically:

\begin{corollary}[Analog of Lemma \ref{lemma-multicastadv-ij-range-3}]\label{corollary-multicastadvc-ij-range-3}
Fix an $(i,j)$-phase in which $i>\lg{n}$, $j<\lg{C}$, and all nodes are active. Fix a node $u$. With probability at least $1-n^{-\Theta(i^2)}$, node $u$ will not become \texttt{helper} by the end of this phase.
\end{corollary}

Essentially, Lemma \ref{lemma-multicastadv-ij-range-1} and Corollary \ref{corollary-multicastadvc-ij-range-3} imply, when $C\leq n/2$, during the execution of \MultiCastAdvC, if all nodes are active, then \texttt{informed} nodes will only become \texttt{helper} in phases in which $i>\lg{n}$ and $j=\lg{C}$.

\bigskip\noindent\textbf{Correctness and competitiveness guarantees.} Firstly, it is easy to verify that Lemma \ref{lemma-multicastadv-helper-imply-inform} still holds. Thus, \MultiCastAdvC also possesses the property that when some node becomes \texttt{helper}, all nodes must have learned the message $m$.

Next, we inspect Lemma \ref{lemma-multicastadv-halt-imply-helper}, which states that when some node \texttt{halt}, all nodes have progressed to at least \texttt{helper} status. In the \MultiCastAdvC setting, the lemma statement still holds, but the proof needs some small adjustments.

\begin{corollary}[Analog of Lemma \ref{lemma-multicastadv-halt-imply-helper}]\label{corollary-multicastadvc-halt-imply-helper}
Fix an $(i,j)$-phase in which all nodes are active, fix a node $u$ that obtained \texttt{helper} status in phase $(\hat{i},\hat{j})$. Assume $i-2/\alpha\geq \hat{i}>\lg{n}$. By the end of step two, the following two events happen simultaneously with probability at most $n^{-\Theta(\hat{i}^2)}$: (a) node $u$ decides to halt; and (b) some node has not progressed to \texttt{helper} status.
\end{corollary}

\begin{proof}[Proof sketch]
We take the same path as in the proof of Lemma \ref{lemma-multicastadv-halt-imply-helper}. Let $\mathcal{E}_1$ be the event that $u$ decides to halt by the end of step two, and $\mathcal{E}_2$ be the event that some node still has not obtained \texttt{helper} status by the end of step two. Let $\mathcal{E}_3$ be the event that ``during step two, for at least 0.02 fraction of slots, Eve jams at least 0.02 fraction of all channels''. We know $\Pr(\mathcal{E}_1\mathcal{E}_2)=\Pr(\mathcal{E}_1\mathcal{E}_2\mathcal{E}_3)+\Pr(\mathcal{E}_1\mathcal{E}_2\overline{\mathcal{E}_3})$.

Bounding $\Pr(\mathcal{E}_1\mathcal{E}_2\mathcal{E}_3)$ is easy. Specifically, this part of the proof does not need any change, thus we can conclude $\Pr(\mathcal{E}_1\mathcal{E}_2\mathcal{E}_3)\leq \Pr(\mathcal{E}_1\mathcal{E}_3)\leq \Pr(\mathcal{E}_1~|~\mathcal{E}_3)\leq n^{-\Theta(i^2)}$.

On the other hand, we bound $\Pr(\mathcal{E}_1\mathcal{E}_2\overline{\mathcal{E}_3})$ via $\Pr(\mathcal{E}_2~|~\mathcal{E}_1\overline{\mathcal{E}_3})$.

Assume $u$ decides to halt by the end of phase $(\tilde{i},\tilde{j})$. Since $u$ obtained \texttt{helper} status in phase $(\hat{i},\hat{j})$, we know $\tilde{i}\geq \hat{i}+2/\alpha$ and $\tilde{j}=\hat{j}$.

Claim \ref{claim-multicastadv-helper-phase-relation-1} and Claim \ref{claim-multicastadv-helper-phase-relation-2} still hold. Therefore, in phase $(\hat{i},\hat{j})$, with probability at least $1-n^{-\Theta(\hat{i}^2)}$, we have $((n-1)/2^{\hat{j}})\cdot(1-p(\hat{i},\hat{j})/2^{\hat{j}})^{n-2}\geq 1.45$ and $(1-p(\hat{i},\hat{j})/2^{\hat{j}})^{n-1}\geq 0.89$. Since $\tilde{i}\geq \hat{i}+2/\alpha$ and $\tilde{j}=\hat{j}$, we also know $p(\tilde{i},\tilde{j})\leq p(\hat{i},\hat{j})/4$. As a result, it is easy to verify, we again have $((n-1)/2^{\tilde{j}})\cdot(1-p(\tilde{i},\tilde{j})/2^{\tilde{j}})^{n-2}>1.59$, and $(1-p(\tilde{i},\tilde{j})/2^{\tilde{j}})^{n-1}>0.94$, with probability at least $1-n^{-\Theta(\hat{i}^2)}$. Assume this is indeed the case.

Besides, by Lemma \ref{lemma-multicastadv-helper-imply-inform}, by the end of phase $(\hat{i},\hat{j})$, all nodes already know $m$, with probability at least $1-n^{-\Theta(\hat{i}^2)}$. Assume indeed all nodes know $m$ in phase $(\tilde{i},\tilde{j})$.

Now, fix a node $v$ that has not obtained \texttt{helper} status at the beginning of phase $(\tilde{i},\tilde{j})$. Let $N^v_m$ and $N^v_s$ denote the number of slots that $v$ hear message $m$ and silence in step two, respectively. Once again, we are able to prove: the probability that $N^v_m<1.5\cdot R(\tilde{i},\tilde{j})\cdot(p(\tilde{i},\tilde{j}))^2$ or $N^v_s<0.9\cdot R(\tilde{i},\tilde{j})\cdot p(\tilde{i},\tilde{j})$ is exponentially small in $R(\tilde{i},\tilde{j})\cdot(p(\tilde{i},\tilde{j}))^2$.

Finally, here comes the part that is different from the proof of Lemma \ref{lemma-multicastadv-halt-imply-helper}. Due to Lemma \ref{lemma-multicastadv-ij-range-1} and Corollary \ref{corollary-multicastadvc-ij-range-3}, it must be the case that $\hat{j}=\tilde{j}=\lg{C}$, with probability at least $1-n^{-\Theta(\hat{i}^2)}$. But in \MultiCastAdvC, in phases in which $j=\lg{C}$, an \texttt{informed} node will become \texttt{helper} so long as $N_m\geq 1.5Rp^2$ and $N_s\geq 0.9Rp$. (This is the exact reason why we remove the $N'_m\leq 2.2Rp^2$ condition when $j=\lg{C}$.) At this point, take a union bound over all $O(n)$ nodes that have not obtained \texttt{helper} status at the beginning of phase $(\tilde{i},\tilde{j})$, we can conclude $\Pr(\mathcal{E}_2~|~\mathcal{E}_1\overline{\mathcal{E}_3})$ is at most $n^{-\Theta(\hat{i}^2)}$.
\end{proof}

\bigskip\noindent\textbf{Fast termination.} Once Eve ceases jamming (more precisely, jamming from Eve is not strong), all nodes should quickly receive the message (if they have not done so already) and then halt. In the \MultiCastAdvC setting, such fast termination properties are also enforced, but exact lemma statements and detailed proofs need some adjustments.

To begin with, the definition of a ``blocking'' epoch needs to be updated so as to capture the fact that Eve has to to jam phases in which $j=\lg{C}$ to effectively disrupt protocol execution.

\begin{definition}\label{def-epoch-blocking-multicastadvc}
Epoch $i$ is \emph{blocking} if at least one of the following two conditions hold: (a) $\mathcal{E}^{> x_1}_{Step1}(> y_1)$ in phase $\lg{C}$; and (b) $\mathcal{E}^{> x_2}_{Step2}(> y_2)$ in phase $\lg{C}$. Here, $x_1=y_1=1/10$, $x_2=y_2=1/10^4$. On the other hand, epoch $i$ is \emph{non-blocking} if both of the following conditions hold: (a) $\mathcal{E}^{\geq 1-x_1}_{Step1}(\leq y_1)$ in phase $\lg{C}$; and (b) $\mathcal{E}^{\geq 1-x_2}_{Step2}(\leq y_2)$ in phase $\lg{C}$.
\end{definition}

With the updated definition, we now state and prove an analog of Lemma \ref{lemma-multicastadv-fast-term-inform}.

\begin{corollary}[Analog of Lemma \ref{lemma-multicastadv-fast-term-inform}]\label{corollary-multicastadvc-fast-term-inform}
Consider phase $\lg{C}$ of an epoch $i>\max\{(\lg(n/(2C)))/\alpha+\lg{C},\lg{n}\}$, assume all nodes are active and are either \texttt{uninformed} or \texttt{informed} at the beginning of this phase. If $\mathcal{E}^{\geq 1-x_1}_{Step1}(\leq y_1)$ happens, then by the end of this phase, all nodes are \texttt{informed}, with probability at least $1-n^{-\Theta(i)}$.
\end{corollary}

\begin{proof}[Proof sketch]
The proof is very similar to the one for Lemma \ref{lemma-multicastadv-fast-term-inform}, except that we use the extra condition $i>(\lg(n/(2C)))/\alpha+\lg{C}$ to ensure $p(i,j)\leq 2C/n$ when $j=\lg{C}$. We now sketch the proof.

Since event $\mathcal{E}^{\geq 0.1}_{Step1}(\leq 0.9)$ occurs, during step one of phase $j=\lg{C}$, there must exist $2d\cdot 2^{2\alpha i-2\alpha j}\cdot i^3$ slots in which Eve jams at most 0.9 fraction of all used channels. Here, $d$ is some sufficiently large constant. Let $\mathcal{R}_1$ be the set of the first half of these $2d\cdot 2^{2\alpha i-2\alpha j}\cdot i^3$ slots, and $\mathcal{R}_2$ be the second half. We further split $\mathcal{R}_1$ into $i>\lg{n}$ segments, each containing $d\cdot 2^{2\alpha i-2\alpha j}\cdot i^2$ slots.

We first focus on slots in $\mathcal{R}_1$, and show that at least $n/2$ nodes will be \texttt{informed} by the end of $\mathcal{R}_1$. To see this, consider some segment $\mathcal{S}$ in $\mathcal{R}_1$, consider some node $u$ that is \texttt{informed} at the beginning of $\mathcal{S}$. Consider a slot in $\mathcal{S}$, let $t\in[1,n/2]$ denote the number of \texttt{informed} nodes at the beginning of this slot. The probability that $u$ informs a previously \texttt{uninformed} node in this slot is at least $p\cdot(1-p/C)^{t-1}\cdot\left(1-(1-p/C)^{n-t}\right)\cdot(1/10)\geq p\cdot(1-p/C)^{n/2}\cdot(1-(1-p/C)^{n/2})\cdot(1/10)\geq p\cdot e^{-np/C}\cdot(1-e^{-np/2C})\cdot(1/10)\geq p\cdot e^{-np/C}\cdot\left(1-(1-np/4C)\right)\cdot(1/10)=(np^2/40C)\cdot e^{-np/C}\geq(p^2/20)\cdot e^{-2}$. Here, the last inequality is due to $C\leq n/2$ and $p\leq 2C/n$. As a result, by the end of $\mathcal{S}$, the probability that no previously \texttt{uninformed} node is informed by $u$ during $\mathcal{S}$ is at most $(1-p^2/(20e^2))^{d\cdot 2^{2\alpha i-2\alpha j}\cdot i^2}\leq e^{-(d/80e^2)\cdot i\cdot\lg{n}}$. Take a union bound over all the $O(n)$ nodes that are \texttt{informed} at the beginning of $\mathcal{S}$, we know the number of \texttt{informed} nodes will at least double by the end of $\mathcal{S}$, with probability at least $1-n\cdot e^{-(d/80e^2)\cdot i\cdot\lg{n}}$. Since there are $i>\lg{n}$ segments in $\mathcal{R}_1$, take another union bound over these segments and we know the number of \texttt{informed} nodes will at least reach $n/2$ by the end of $\mathcal{R}_1$, with probability at least $1-i\cdot n\cdot e^{-(d/80e^2)\cdot i\cdot\lg{n}}$. For sufficiently large $d$, this probability is at least $1-n^{-\Theta(i)}$.

We now turn our attention to the slots in $\mathcal{R}_2$. Due to the above analysis, assume at the beginning of $\mathcal{R}_2$, at least $n/2$ nodes are \texttt{informed}. Consider a slot in $\mathcal{R}_2$ and a node $u$ that is still \texttt{uninformed} at the beginning of this slot. Let $t\in[n/2,n)$ denote the number of \texttt{informed} nodes at the beginning of this slot. The probability that $u$ will hear message $m$ in this slot is at least $p\cdot t\cdot (p/C)\cdot\left(1-p/C\right)^{t-1}\cdot(1/10)\geq p\cdot(n/2)\cdot(p/C)\cdot\left(1-p/C\right)^{n}\cdot(1/10)\geq p\cdot(n/2)\cdot(p/C)\cdot e^{-2np/C}\cdot(1/10)\geq (p^2/10)\cdot e^{-4}$. Here, the last inequality is due to $C\leq n/2$ and $p\leq 2C/n$. As a result, by the end of $\mathcal{R}_2$, the probability that $u$ is still \texttt{uninformed} is at most $(1-(p^2/10)\cdot e^{-4})^{d\cdot 2^{2\alpha i-2\alpha j}\cdot i^3}\leq e^{-(d/40e^4)\cdot i^2\cdot \lg{n}}$. Take a union bound over the at most $n/2$ nodes that are \texttt{uninformed} at the beginning of $\mathcal{R}_2$, we know all nodes will be \texttt{informed} by the end of $\mathcal{R}_2$, with probability at least $1-n\cdot e^{-(d/40e^4)\cdot i^2\cdot \lg{n}}$. For sufficiently large $d$, this probability is at least $1-n^{-\Theta(i^2)}$.
\end{proof}

Next, we prove an analog of Lemma \ref{lemma-multicastadv-fast-term-helper}, showing that when all nodes are \texttt{informed} or \texttt{helper}, after a non-blocking epoch, all nodes must have progressed to \texttt{helper} or \texttt{halt} status.

\begin{corollary}[Analog of Lemma \ref{lemma-multicastadv-fast-term-helper}]\label{corollary-multicastadvc-fast-term-helper}
Consider phase $\lg{C}$ of an epoch $i\geq\max\{(\lg(32n/C))/\alpha+\lg{C},\lg{n}\}$, assume all nodes are active and are either \texttt{informed} or \texttt{helper} at the beginning of step two of this phase. If $\mathcal{E}^{\geq 1-x_2}_{Step2}(\leq y_2)$ happens, then by the end of this phase, all nodes must be in \texttt{helper} or \texttt{halt} status, with probability at least $1-n^{-\Theta(i^2)}$.
\end{corollary}

\begin{proof}[Proof sketch]
The proof is very similar to the one for Lemma \ref{lemma-multicastadv-fast-term-helper}, except that we use the extra condition $i>(\lg(32n/C))/\alpha+\lg{C}$ to ensure $p(i,j)\leq C/(32n)$ when $j=\lg{C}$.

Notice that to prove the lemma, we only need to show by the end of phase $\lg{C}$, all nodes have \emph{at least} progressed to \texttt{helper} status, with sufficiently high probability.

Consider a node $u$ that is only \texttt{informed} at the beginning of step two of phase $\lg{C}$. Since event $\mathcal{E}^{\geq 0.9999}_{Step2}(\leq 0.0001)$ occurs, event $\mathcal{E}^{\geq 0.99}_{Step2}(\leq 0.01)$ must occur. Now, in order to obtain the smallest possible $N_m$ and $N_s$ for node $u$, without loss of generality, assume during step two, Eve jams 0.01 fraction of all channels in 0.99 fraction of slots, and she jams all channels in remaining slots. Let $\mathcal{R}_1$ denote the set of slots in which Eve does not jam all channels, we know $\mathbb{E}[N_m]=|\mathcal{R}_1|\cdot p\cdot 0.99\cdot (n-1)\cdot(p/C)\cdot(1-p/C)^{n-2}\geq 0.99R\cdot p\cdot 0.99\cdot 0.99n\cdot(p/C)\cdot e^{-2np/C}\geq 0.99^3\cdot 2Rp^2\cdot e^{-1/16}>1.8Rp^2$. Here, the second to last inequality is due to $C\leq n/2$ and $p\leq C/32n$. Similarly, $\mathbb{E}[N_s]=|\mathcal{R}_1|\cdot p\cdot 0.99\cdot (1-p/C)^{n-1}\geq 0.99R\cdot p\cdot 0.99\cdot e^{-2np/C}\geq 0.99^2\cdot Rp\cdot e^{-1/16}>0.92Rp$. Since whether $u$ will hear message $m$ (or silence) are independent among slots in $\mathcal{R}_1$, by a Chernoff bound, the probability that $N_m$ is less than $1.5Rp^2$ or $N_s$ is less than $0.9Rp$ is at most $n^{-\Theta(i^2)}$.

Take a union bound over all the $O(n)$ \texttt{informed} nodes at the beginning of step two of phase $\lg{C}$, we know all nodes will at least progress to \texttt{helper} status by the end of phase $\lg{C}$, with probability at least $1-n^{-\Theta(i^2)}$.
\end{proof}

Lastly, Lemma \ref{lemma-multicastadv-fast-term-halt} shows, once a \texttt{helper} starts considering termination, after $O(1)$ epochs, the node will successfully halt in a non-blocking epoch. In the \MultiCastAdvC setting, it still holds.

\bigskip\noindent\textbf{Proof of the main theorem.} We are now ready to sketch the proof for Theorem \ref{thm-multicastadvc}.

To begin with, the following nice properties still hold in the \MultiCastAdvC setting: (a) due to Lemma \ref{lemma-multicastadv-ij-range-1} and Lemma \ref{lemma-multicastadv-helper-imply-inform}, when the first \texttt{helper} appears, all nodes must have already learned the message $m$, with high probability; (b) due to Lemma \ref{lemma-multicastadv-ij-range-1} and Corollary \ref{corollary-multicastadvc-halt-imply-helper}, when the first \texttt{halt} node appears, all nodes must have already progressed to \texttt{helper} status, with high probability; and (c) due to Lemma \ref{lemma-multicastadv-ij-range-1}, Corollary \ref{corollary-multicastadvc-ij-range-3}, and Corollary \ref{corollary-multicastadvc-halt-imply-helper}, for each node $u$, if it obtained \texttt{helper} status in phase $(\hat{i}_u,\hat{j}_u)$, then it must be the case that $\hat{i}_u>\lg{n}$ and $\hat{j}_u=\lg{C}$, with high probability.

In the following proof, assume the above properties indeed hold.

Let $l$ denote the last epoch that is blocking while some node is still active at the beginning of that epoch. By the end of epoch $l$, we know the status of all nodes must belong to exactly one of the following four cases: (1) all nodes are active and either \texttt{uninformed} or \texttt{informed} (and there exists at least one \texttt{uninformed} node); (2) all nodes are active and either \texttt{informed} or \texttt{helper} (and there exists at least one \texttt{informed} node); (3) every node is either \texttt{helper} or has terminated (and there exists at least one \texttt{helper} node); or (4) all nodes have terminated.

Let $\beta=\max\{(\lg(32n/C))/\alpha+\lg{C},\lg{n}\}$, following analysis consider two complement scenarios: $l\geq\beta$ and $l<\beta$.

\textbf{The $l\geq\beta$ scenario.} In this situation, due to Corollary \ref{corollary-multicastadvc-fast-term-inform}, Corollary \ref{corollary-multicastadvc-fast-term-helper}, and Lemma \ref{lemma-multicastadv-fast-term-halt}, it is easy to verify, no matter which case (among the aforementioned four cases) the system is in by the end of epoch $l$, all nodes must have terminated by the end of epoch $l+11/\alpha+2\leq l+12/\alpha$, w.h.p.

To bound the runtime and cost of honest nodes, we need to bound $l$. Since epoch $l$ is blocking, $\mathcal{E}^{> x_1}_{Step1}(> y_1)$ or $\mathcal{E}^{> x_2}_{Step2}(> y_2)$ must have occurred during phase $\lg{C}$ of this epoch. Thus, the cost of Eve during epoch $l$ is at least $b\cdot x_2y_2\cdot 2^{\alpha(2l-2\lg{C})}\cdot l^3\cdot 2^{\lg{C}} = b\cdot x_2y_2\cdot 2^{2\alpha l}\cdot l^3\cdot 2^{(1-2\alpha)\lg{C}}$, implying $T\geq x_2y_2\cdot 2^{2\alpha l}\cdot l^3\cdot C^{(1-2\alpha)}\geq x_2y_2\cdot 2^{2\alpha l}\cdot C^{(1-2\alpha)}$. Hence, $l\leq (\lg(T/(x_2y_2\cdot C^{(1-2\alpha)})))/(2\alpha)$.

At this point, we can conclude each node will halt within:

\vspace{-3ex}
\begin{align*}
& \sum_{i=1}^{l+12/\alpha}\sum_{j=0}^{i-1}\Theta(1)\cdot 2^{2\alpha i}\cdot 2^{-2\alpha j}\cdot i^3\leq \Theta\left(\frac{1}{1-2^{-2\alpha}}\right)\cdot \sum_{i=1}^{l+12/\alpha}2^{2\alpha i}\cdot i^3\\
\leq & \Theta\left(\frac{1}{1-2^{-2\alpha}}\right)\cdot\left(l+\frac{12}{\alpha}\right)^3\cdot \sum_{i=1}^{l+12/\alpha}2^{2\alpha i}\\
\leq & \Theta\left(\frac{1}{1-2^{-2\alpha}}\right)\cdot(2l)^3\cdot\frac{2^{2\alpha}\cdot 2^{2\alpha(l+12/\alpha)}}{2^{2\alpha}-1}\\
\leq & \Theta\left(\frac{1}{1-2^{-2\alpha}}\right)\cdot\left(\frac{2}{2\alpha}\lg{T}\right)^3\cdot\frac{2^{2\alpha+24}}{2^{2\alpha}-1}\cdot\frac{T}{x_2y_2\cdot C^{(1-2\alpha)}}\\
= & \Theta(1)\cdot \frac{T}{C^{(1-2\alpha)}}\cdot\lg^3{T}
\end{align*}

To bound each node's cost, first consider epochs one to $\lg{n}$. For each node, the total cost of it during these epochs cannot exceed the total length of these epochs. Thus, for each node, the total cost during the first $\lg{n}$ epochs is at most:

\vspace{-2ex}
$$\sum_{i=1}^{\lg{n}}\sum_{j=0}^{i-1}\Theta(1)\cdot 2^{2\alpha i}\cdot 2^{-2\alpha j}\cdot i^3 \leq \Theta\left(n^{2\alpha}\cdot\lg^3{n}\right)$$

Starting from epoch $\lg{n}+1$, each node's cost within an epoch will be tightly concentrated around its expectation, with sufficiently high probability. Thus, with high probability, for every node, the total cost after epoch $\lg{n}$ is at most:

\vspace{-3ex}
\begin{align*}
& \sum_{i=1}^{l+12/\alpha}\sum_{j=0}^{i-1}\Theta(1)\cdot 2^{\alpha i}\cdot 2^{-\alpha j}\cdot i^3 \leq \Theta\left(\frac{1}{1-2^{-\alpha}}\right)\cdot \sum_{i=1}^{l+12/\alpha}2^{\alpha i}\cdot i^3\\
\leq & \Theta\left(\frac{1}{1-2^{-\alpha}}\right)\cdot\left(l+\frac{12}{\alpha}\right)^3\cdot \sum_{i=1}^{l+12/\alpha}2^{\alpha i}\\
\leq & \Theta\left(\frac{1}{1-2^{-\alpha}}\right)\cdot(2l)^3\cdot\frac{2^{\alpha}\cdot 2^{\alpha(l+12/\alpha)}}{2^{\alpha}-1}\\
\leq & \Theta\left(\frac{1}{1-2^{-\alpha}}\right)\cdot\left(\frac{2}{2\alpha}\lg{T}\right)^3\cdot\frac{2^{\alpha+12}}{2^{\alpha}-1}\cdot\left(\frac{T}{x_2y_2\cdot C^{(1-2\alpha)}}\right)^{1/2}\\
= & \Theta(1)\cdot\left(\frac{T}{C^{(1-2\alpha)}}\right)^{1/2}\cdot\lg^3{T}
\end{align*}

As a result, when $l\geq\beta$, with high probability, the total cost of each node is $O((T/(C^{(1-2\alpha)}))^{1/2}\cdot\lg^3{T}+n^{2\alpha}\cdot\lg^3{n})$.

\textbf{The $l<\beta$ scenario.} This situation is easier to analyze. By an analysis similar to the $l\geq\beta$ scenario, we know all nodes will terminate by the end of epoch $\kappa=\beta+11/\alpha+2\leq \lg{n}+(\lg(32n/C))/\alpha+\lg{C}+12/\alpha$, with high probability. Therefore, the total runtime and cost for each node in this situation is at most:

\vspace{-2ex}
\begin{align*}
& \sum_{i=1}^{\kappa}\sum_{j=0}^{i-1}\Theta(1)\cdot 2^{2\alpha i}\cdot 2^{-2\alpha j}\cdot i^3 \leq \Theta\left(\frac{1}{1-2^{-2\alpha}}\right)\cdot \sum_{i=1}^{\kappa}2^{2\alpha i}\cdot i^3\\
\leq & \Theta\left(\frac{1}{1-2^{-2\alpha}}\right)\cdot\left(\frac{4\lg{n}}{\alpha}\right)^3\cdot\sum_{i=1}^{\kappa}2^{2\alpha i}\\
\leq & \Theta\left(\frac{1}{1-2^{-2\alpha}}\right)\cdot\left(\frac{4\lg{n}}{\alpha}\right)^3\cdot\frac{2^{2\alpha}\cdot 2^{2\alpha\kappa}}{2^{2\alpha}-1}\\
= & \Theta\left(n^{2+2\alpha}/C^{2-2\alpha}\cdot\lg^3{n}\right)
\end{align*}

We conclude the proof by noting that the $l<\beta$ scenario also includes the situation in which Eve is not present.

\end{document}